\documentclass[runningheads]{llncs}

\usepackage[T1]{fontenc}
\usepackage{graphicx}
\usepackage{amssymb}
\usepackage{amsmath}
\usepackage{amsfonts}
\usepackage{xcolor}
\usepackage{comment}
\usepackage[disable]{todonotes}
\usepackage{subcaption}
\usepackage{stmaryrd}
\usepackage{hyperref}
\usepackage{float} 
\usepackage{fp}
\usepackage[createShortEnv]{proof-at-the-end}
\usepackage{scalefnt}
\usepackage{lmodern}

\usetikzlibrary{patterns}
\tikzset{every picture/.style={line width=0.75pt}} 

\newcommand{\emptycell}[4]{
    \tikzstyle{roundnode}=[circle,draw = #3,fill=white, minimum size = 56];
    \node[roundnode] (#4) at (#1,#2) {};
    \draw[color = #3] (#1,#2 -1) -- (#1,#2 +1);
}

\newcommand{\leftmoover}[4]{
    \tikzstyle{roundnode}=[circle,draw = #3,fill=white, minimum size = 56];
    \node[roundnode] (#4) at (#1,#2) {};
    \draw[color = #3] (#1,#2 -1) -- (#1,#2 +1);
    \begin{scope}
        \clip (#1,#2-1) rectangle (#1-1,#2+1);
        \draw[color = #3][fill = #3] (#1,#2) circle(1);
    \end{scope}
}

\newcommand{\rightmoover}[4]{
    \tikzstyle{roundnode}=[circle,draw = #3,fill=white, minimum size = 56];
    \node[roundnode] (#4) at (#1,#2) {};
    \draw[color = #3] (#1,#2 -1) -- (#1,#2 +1);
    \begin{scope}
        \clip (#1,#2-1) rectangle (#1+1,#2+1);
        \draw[color = #3][fill = #3] (#1,#2) circle(1);
    \end{scope}
}

\newcommand{\rightleft}[4]{
    \tikzstyle{roundnode}=[circle,draw = #3,fill=#3, minimum size = 56];
    \node[roundnode] (#4) at (#1,#2) {};
}

\newcommand{\vertex}[6]{
    \ifnum #4=0
    {\emptycell{#1}{#2}{#3}{#6}}
    \fi
    \ifnum #4=1
    {\rightmoover{#1}{#2}{#3}{#6}}
    \fi
    \ifnum #4=2
    {\leftmoover{#1}{#2}{#3}{#6}}
    \fi
    \ifnum #4=3
    {\rightleft{#1}{#2}{#3}{#6}}
    \fi
    \draw (#1+1,#2) node[right]{{\Huge \color{#3} #5}};
}

\renewcommand{\sp}[1]{\lfloor #1\rfloor}
\newcommand{\renameclass}[1]{[#1]}
\newcommand{\namealg}[1]{\widehat{#1}}
\newcommand{\nameeq}{\,\hat{=}\,}
\newcommand{\namesubseteq}{\,\hat{\subseteq}\,}
\newcommand{\namecap}{\,\hat{\cap}\,}
\newcommand{\restr}[2]{\mathcal{M}_{#1}(#2)}
\newcommand{\restri}[1]{{\mathcal{M}_{#1}}}
\newcommand{\restriD}[2]{{\mathcal{M}^{#1}_{#2}}}
\newcommand{\restrD}[3]{{\mathcal{M}^{#1}_{#2}}(#3)}

\newcommand{\crestri}[1]{{{\overline{\mathcal{M}}}_{#1}}}
\newcommand{\restrA}[2]{\mathcal{M}_{#1}(#2)}
\newcommand{\crestrA}[2]{{\overline{\mathcal{M}}}_{#1}(#2)}
\newcommand{\restriA}[1]{{\mathcal{M}_{#1}}}
\newcommand{\crestriA}[1]{{{\overline{\mathcal{M}}}_{#1}}}
\newcommand{\restriB}[1]{{\mathcal{N}_{#1}}}
\newcommand{\crestriB}[1]{{{\overline{\mathcal{N}}}_{#1}}}
\newcommand{\restrB}[2]{\mathcal{N}_{#1}(#2)}

\newcommand{\valid}[2]{\Omega_{#1}(#2)}
\newcommand{\graphvalidexplicit}[2]{\Gamma_{{#1}_{#2}}}
\newcommand{\graphvalid}[2]{\graphvalidexplicit{#1}{#2}}
\newcommand{\graphvalidmutex}[2]{\graphvalidexplicit{\restriD{#1}{#2}}{}}

\newcommand{\upast}[1]{{\cal P}_{#1}}
\newcommand{\ufut}[1]{\mathcal{F}_{#1}}
\newcommand{\transpose}[1]{{#1}^{\mathrm{T}}}

\newcommand{\sshift}{\mathcal{S}}
\newcommand{\position}{position}

\newcommand{\Past}{\mathrm{Past}}
\newcommand{\Fut}{\mathrm{Fut}}
\newcommand{\I}{\mathrm{I}}
\newcommand{\B}{\mathrm{B}}
\newcommand{\V}{\mathrm{V}}

\newcommand{\D}{\mathcal{D}}

\newcommand{\E}{\mathrm{E}}

\renewcommand{\restriction}{\mathrel{\!\raisebox{-.5ex}{$\vert$}}}

\newcommand{\VL}[1]{}
\newcommand{\withproofs}[1]{}

\newcommand{\calD}{{\ensuremath{\mathcal{D}}}}
\newcommand{\calE}{{\ensuremath{\mathcal{E}}}}
\newcommand{\calG}{{\ensuremath{\mathcal{G}}}}

\newcommand{\calP}{{\ensuremath{\mathcal{P}}}}

\newcommand{\calV}{{\ensuremath{\mathcal{V}}}}

\newcommand{\calX}{{\ensuremath{\mathcal{X}}}}
\newcommand{\X}{\calX}

\newcommand{\maxi}[0]{^{\infty}}
\newcommand{\maxgraph}[0]{fully-explored }
\newcommand{\maximality}[0]{forward full-exploration }
\newcommand{\bckwmaximality}[0]{backward full-exploration }
\newcommand{\xCone}[0]{\mathcal{C}_x}
\newcommand{\xcone}[1]{\mathcal{C}_{#1}}
\newcommand{\diskA}[1]{\sshift_{\restriA{#1}}}
\newcommand{\diskB}[1]{\sshift_{\restriB{#1}}}

\definecolor{lowgreen}{rgb}{0.40390625,0.6109375,0.09109375}
\definecolor{myblack}{RGB}{0,0,0}
\definecolor{border}{RGB}{206,206,206}
\definecolor{port}{RGB}{155,155,155}
\definecolor{setBorder}{RGB}{80,227,194}
\definecolor{internal}{RGB}{38,105,185}
\definecolor{greenstate}{RGB}{80,183,100}
\definecolor{greenstategray}{RGB}{158,220,170}
\definecolor{redstate}{RGB}{201,109,76}
\DeclareRobustCommand{\sizeFtwo}[0]{\Large} 
\DeclareRobustCommand{\sizeFtwop}[0]{\normalsize} 

\newcommand{\MC}[1]{{\color{cyan}#1}}
\renewcommand{\MC}[1]{{#1}}
\newcommand{\PA}[1]{{\color{magenta}#1}}
\newcommand{\RX}[2]{{#1\todo{WARNING}}}{}
\renewcommand{\PA}[1]{{#1}}

\begin{document}
\title{Space-time reversible graph rewriting}
%
%
%
%
\author{Pablo Arrighi\inst{1}\orcidID{0000-0002-3535-1009} \and
Marin Costes\inst{1}\orcidID{0000-0003-0915-9192} \and
Luidnel Maignan\inst{2}\orcidID{0009-0009-5297-5022}}
\authorrunning{P. Arrighi et al.}
%
\institute{Université Paris-Saclay, Inria, CNRS, LMF, 91190 Gif-sur-Yvette, France \and
Univ Paris Est Creteil, LACL, 94000, Creteil, France}
\maketitle              
\begin{abstract}
In the mathematical tradition, reversibility requires that the evolution of a dynamical system be a bijective function. In the context of graph rewriting, however, the evolution is not even a function, because it is not even deterministic---as the rewrite rules get applied at non-deterministically chosen locations. 
Physics, by contrast, suggests a more flexible understanding of reversibility in space-time, whereby any two closeby snapshots (aka `space-like cuts'), must mutually determine each other.
We build upon the recently developed framework of space-time deterministic graph rewriting, in order to formalise this notion of space-time reversibility, and henceforth study reversible graph rewriting.  
We establish sufficient, local conditions on the rewrite rules so that they be space-time reversible. 
We provide an example 
featuring time dilation, in the spirit of general relativity.

\end{abstract}
%
%
%
\todo[inline]{Marin: Uncolour Pablo's edits. More diagrams for pedagogy.}
\todo[inline,color=green]{Marin: Check relax node creation.}
\todo[inline,color=red]{Luidnel: Tout relire avec un oeuil novice.\\ Check that it's self-contained. That every notation is defined.\\ Mark where an extra figure could help understand a complicated notion.\\ Check plural vs singular uses.}

\section{Introduction}


{\em Dynamical systems from grids to graphs}.
Dynamical systems refer to the evolution of an entire configuration, seen as a monolithic global state at time $t$, into another configuration at time $t+1$, and then $t+2$, etc., iteratively. When the global state is not really monolithic but rather a compound of local systems, the global function is often described as a composition of local functions acting across space. 
This is famously the case for grid-based dynamical systems such as cellular automata~\cite{Hedlund}, but also for their graph-based extension, namely causal graph dynamics~\cite{ArrighiCGD,ArrighiCayley}. These were designed to model distributed systems whose the interconnection network is dynamical, e.g. social networks, biological systems, or physical systems such as discretised general relativity~\cite{Sorkin}.
Recently, causal graph dynamics have then been extended in two directions which this paper aims to merge: reversibility on the one hand, and asynchronism on the other.

{\em Reversibility} is one of the most important property of dynamical systems. It refers to the requirement that the global function be a bijection over the set of configurations---and that the local functions composing it be themselves bijective.
Reversibility shows up in various aspects of Computer Science. It holds the key to diminishing power consumption~\cite{Landauer}, which has become so critical. It is useful for debugging~\cite{ReversibleDebugging} and more generally for the reproducibility of system behaviour. In distributed computing, it is useful as a failure handling primitive, e.g. for rolling back a transaction~\cite{ReversibleRollback}. In natural computing, it may capture the features of reversible chemical/biological reactions~\cite{ReversibleReactions}. In quantum computing, unitarity entails reversibility: the study of the reversible version of a model of computation is quite often a sound prior step to take before moving to its quantum version. Both cellular automata and causal graph dynamics have followed this path, i.e. from reversible cellular automata~\cite{KariBlock} and reversible causal graph dynamics~\cite{ArrighiRCGD} to their quantum counterpart~\cite{ArrighiOverview,ArrighiQNT}.
For reversible causal graph dynamics, one of the challenges was to handle the creation/destruction of nodes whilst remaining reversible and renaming-invariant. The issue was overcome by introducing algebraic operations upon the names of nodes, akin to splitting/merging~\cite{ArrighiCreation}.


{\em Asynchronism.} The assumption of synchronism that underlies dynamical systems is often criticised: distributed computation is fundamentally asynchronous and Physics itself departs from the idea of a global time across the universe. Asynchronism---the application of local operators at arbitrary places, non-determinis\-tically---fits both these pictures better and is also well studied. Still it is often the case that, in spite of non-determinism, some form of well-definiteness of events must be preserved. Space-time deterministic graph rewriting~\cite{ArrighiDetRew} identifies a way to relax the synchronism of causal graph dynamics to allow all possible non-deterministic scheduling and yet lead to a unique unfolding of events---a unique space-time diagram in the sense of cellular automata or physics. The challenge there was to identify simple enough local conditions, to ensure the existence of a consistent space-time diagram at the global level.


{\em Reversiblity vs Asynchronism.} 
On the face of it reversibility and asynchronism are incompatible. Indeed, asynchronism leads to non-determinism, but then how can we demand that the evolution be a bijective function when it is not even a function? Under space-time determinism~\cite{ArrighiDetRew}, however, asynchronism leads to a non-determinism of scheduling which \emph{does not really matter} as far as space-time is concerned. The question arises, therefore, whether this scheduling non-determinism which we already overcame to recover determinism at the space-time level, can again be overcome to recover a meaningful notion of `space-time reversibility'.
From a general philosophy of science point of view, we find it compelling to formally reconcile asynchronism and reversibility.
These considerations also have applications in Physics, e.g. for the sake of  mathematically sound, constructive frameworks for discrete models of general relativity~\cite{Sorkin}. 

{\em Contributions.} This works aims at coming up with a sensible notion of reversibility in the presence of asynchronism. It does so within the space-time deterministic graph rewriting framework, using the lessons taking from reversible causal graph dynamics. Our core theoretical contributions are: 1/ a rigorous formalisation of space-time reversibility 2/ the provision of an axiomatic characterization of this reversibility---whereby a limited number of high-level conditions need be checked 3/ the provision of a constructive characterization of reversibility---in terms of concrete input/output patterns in the style of rewrite systems.\\ 
{\em Plan.} Prior work is necessary to get there: extending the name algebra of~\cite{ArrighiCreation} (Sec.~\ref{sec:names}); recalling the notions of graphs and locality that we use, and establishing a first formal connection with rewrite systems (Sec.~\ref{sec:graphs}). Once this is done we provide our three characterizations of reversibility (Sec.~\ref{sec:rev}) and achieve space-time determinism of the inverse (Sec.~\ref{sec:commutation}). Finally, we provide a complete example featuring time-dilation reversibility, renaming-invariance (Sec.~\ref{sec:timedilation}) before we conclude (Sec.~\ref{sec:conclusion}).


\section{An algebra for naming vertices}\label{sec:names}

In a variety of different early formalisms, it was shown that reversibility and causality leads to vertex-preservation, i.e. the forbidding of vertex creation and destruction~\cite{ArrighiRCGD}. 
This limitation was finally overcome by introducing a \textit{name algebra}~\cite{ArrighiCreation}. Let us quickly remind the reader of why we cannot do without such an algebra. 
Say that some reversible evolution splits a node named $x$ into two nodes. We need to name the two infants in a way that avoids name conflicts with the vertices of the rest of the graph. But if the evolution is locally-causal, we are unable to just `pick a fresh name out of the blue', because we do not know which names are available. Thus, we have to construct new names locally. A natural choice is to use $x.l$ and $x.r$ (for left and right respectively).
But then, reversibility forces us to allow for the merger of two node names as the correct inverse operation of node splitting.
We are therefore compelled to accept that vertex names obey such algebraic rules as $x.l\lor x.r=x$. We must also disallow that one node be called $x$ and another be called $x.l$, because this could cause a name conflict e.g. if $x$ splits into $x.l$ and $x.r$.

In the context of  space-time deterministic graph rewriting, however, each vertex must actually be understood as a computational process at position $x$ and time tag $t$, hence carrying a time tag~\cite{ArrighiDetRew}. This time tag is essential if we want to be able to distinguish individual events and demand that they be well-determined. We must therefore understand how time tags interact with the position algebra. It turns out that adding a few extra algebraic rules suffices to enforce commutation between time increments one the one hand, and name splitting/merging on the other
. We reach:
\begin{definition}[Name algebra]\label{def:namealgebra}
The {\em name algebra} $\calV$ is defined as the terms given by the grammar 
\begin{equation}
u,v\ ::=\ m\ |\ u.p\ |\ u\lor v\ |\  t.u\quad\text{with}\quad m\in \mathbb{N},\ p\in \{l,r\}^{*},\ t\in \mathbb{Z}
\label{eq:algebra}
\end{equation}
and endowed with the following equality theory (with $ \varepsilon $ the empty word):
\begin{align*}
u.\varepsilon&=u  & u.(p\cdot q)&=(u.p).q   &  u.l\lor u.r&=u  \\  
(u\lor v).l&=u   &  (u\lor v).r&=v  &  t.(u.p)&=(t.u).p \\
t.(u\lor v)&=t.u\lor t.v   &  s.(t.u)&=(s+t).u  
 &   0.u&=u.
\end{align*}
The {\em \position\ algebra} $\X$ is defined as the subset of those terms of $\calV$ obtained without using the last case of Eq.~\eqref{eq:algebra}.
Let  $\sp{\cdot}:\calV \to \X$ be the projector on the \position s inductively as follows:
\begin{align*}
\sp{m}&=m & \sp{u.p}&=\sp{u}.p & \sp{u\lor v}&=\sp{u}\lor \sp{v} & \sp{t.u}&=\sp{u}
\end{align*}
Consider $U\subseteq\calV$. The {\em closure} $\namealg{U}$ is defined as the smallest subalgebra of $\calV$ (i.e. closed under the last three operations of Eq. \eqref{def:namealgebra}) that contains $U$.\\ 
We use letters $m,n$ to designate elements of $\mathbb{N}$; $x,y$ to designate \position s in $\X$;  
$s,t$ to designate time tags in $\mathbb{Z}$; and $u,v$ to designate names in $\calV$.\\
We use $U\nameeq U'$ as a shorthand notation for $\namealg{U}=\namealg{U}'$, $U\namecap U'$ as a shorthand notation for $\namealg{U}\cap\namealg{U}'\neq \varnothing$, and  $U\namesubseteq U'$ as a shorthand notation for $\namealg{U}\subseteq\namealg{U}'$. 
\end{definition}

\noindent Notice that giving a simple name such as $u$ is not that innocuous at this stage. Indeed, the graphs may not contain vertex $u$, but it may contain $u.l$ and $u.r\lor w$ say, and these two may well lie far apart in the graph---in general $\{v\in V_G \mid u\namecap v\}$ could be quite large and spread out. 

Still, in the end the names of the vertices are just intended to describe the geometry, and nothing else. Thus, the kind of operators (neighbourhoods, local operators) that we will consider will typically required to be renaming-invariant---capturing the idea that no matter where we are in space-time, the same causes lead to the same effects. 
\begin{definition}[Renaming and renaming-invariance]\label{def:renaminginvariance}
A renaming is a function $R:\calV\rightarrow \calV$ such that
\begin{equation*}
R(u.p) =R(u) .p\qquad R(u\lor v) =R(u) \lor R(v)\qquad R(t.u)=t.R(u)
\end{equation*}
and verifying that $\sp{R(.)}:\X\rightarrow \X$ is a bijection. It is fully specified by its action on domain $ \mathbb{N}$. It is extended to act upon graphs by renaming their nodes.\\
Let $F$ be a function over graphs, possibly parameterized by positions $x \in \X$.
It said to be renaming-invariant if and only if $RF_{x} =F_{\sp{R(x)}} R$.
\end{definition}

\section{Graphs and locality}\label{sec:graphs}

\subsection{A high-level specification of the dynamics}

Let us start by formally introducing the type of graphs that we consider: directed acyclic labelled port graphs. By ``port graph'' we mean that edges are attached to the ports of the nodes, rather than the nodes themselves. We will fix $\pi=\{a,b,\ldots\}$ a finite set of ports, and write $u\!:\!a$ to designate port $:\!a$ of node $u$. By ``labelled graph'' we mean that node $u$ carries an internal state $\sigma(u)$ taken to belong to a finite set $\Sigma=\{0,1,\ldots\}$.
The use of labelled port graphs is totally standard in distributed computing~\cite{PapazianRemila,Chalopin,KrivineKappa}. This is because ports are mandatory in order to be able to tell a neighbouring process from another, whereas labels are required in order to capture the state of each process. The use of a DAG is also quite common in order to capture the dependency between the processes~\cite{CausalSets,AsynchronousSimsSync}. The merger of DAG and port graphs is less common and well-argued in~\cite{ArrighiDetRew}. In addition, because we take partial views of these graphs, they will possibly have ``borders'', i.e. dangling edges from/to internal vertices to/from border vertices---whose internal states are unknown. We reach the following definition, as illustrated by Fig.~\ref{fig : Induced subgraph 1}\todo[inline,color=gray]{To Marin: Referees want a better examples e.g. illustrating more conditions.}.
\begin{definition}[Graphs]\label{def : graphs}
A \emph{graph} $G$ is given by a tuple $(\I_G, \B_G, \E_G, \sigma_G)$ where:
\begin{itemize}
    \item $\I_G \subseteq \calV$ is the set of \emph{internal vertices of $G$},
    \item $\B_G \subseteq \calV \setminus \I_G$ is the set of \emph{border vertices of $G$},
    \item $\E_G \subseteq (V_G\!:\!\pi )^{2} \setminus ( \B_{G}\!:\!\pi )^2$ is the set of \emph{(oriented) edges}, and
    \item $\sigma_G :\I_G\rightarrow \Sigma$ maps each internal vertex to its state.
\end{itemize}
where we denote by $\V_G := \I_G\cup \B_G$ the set of all vertices of the graph, and by $(V\!:\!\pi):=\{v\!:\!p\mid v\in V,\,p\in\pi\}$ the set of ports of some set of vertices $V$. Moreover the tuple has to be such that:
\begin{description} 
    \item{acyclicity:} 
    the graph has no cycles.
    \todo[inline,color=gray]{To Marin: we are not asking for any finiteness, well-foundedness, OK to have infinte chains? Marin : Maybe we can ask to only have finite chains (cf Cauchy surface). (if we ask well foundness, then it is not stable by transposition)}
    \item{border-attachment:} $\forall u \in \B_G,\, \exists (v\!:\!a,v'\!:\!a') \in \E_G,\, u \in \{\,v,v'\,\}$,\\ i.e. border vertices lie at distance one from internal vertices.
    \item{port-saturation:} $\forall (u\!:\!a,v\!:\!b),(u'\!:\!a',v'\!:\!b') \in {\E_G},\, u\!:\!a\neq v'\!:\!b' \,\wedge\, u\!:\!a=u'\!:\!a' \Leftrightarrow v\!:\!b=v'\!:\!b'$\\
    i.e. ports are used only once---so as to distinguish each neighbour.
    \item{non-overlapping positions:} $\forall t.x, t'.x' \in {\V_{G}},\,\forall p,p'\in\{l,r\}^*,\, x.p = x'.p' \Rightarrow t.x = t'.x'$,\\
    i.e. position fragments appear only once---so as to avoid name conflicts.
\end{description}
We denote by $\Past(G)\subseteq I_G$ the vertices of $I_G$ with no incoming edges, and by $\Fut(G)\subseteq I_G$ those with no outgoing edges, \textit{i.e.}, $\Fut(G) = \Past(\transpose{G})$ where $\transpose{G}$ is the transposed graph.\\
We denote by $\mathcal{G}$ the set of all graphs. Given a set of graphs $\mathcal{S}$, we denote by $\mathcal{S}\maxi=\{G\in \mathcal{S}\,|\, \forall (u\!:\!a,v\!:\!b)\in E_G,\,v\in I_G \}$ its subset of \maxgraph graphs, and by ${}\maxi\!\mathcal{S}=\{G\in \mathcal{S}\,|\, \forall (u\!:\!a,v\!:\!b)\in E_G,\,u\in I_G \}$ those who are backwards \maxgraph\!.
\end{definition}
Intuitively, vertices represent computational processes and each edge expresses that the target process is awaiting for the source process. The $\Past(G)$ vertices stand for processes that are no longer awaiting for results by others, and are therefore ready to be executed.

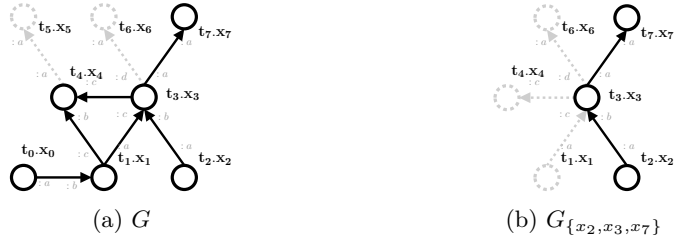
\begin{figure}[t]
\hfil
\begin{subfigure}{0.25\textwidth}
    \captionsetup{justification=centering}
    \centering
    \resizebox{\textwidth}{!}{\tikzset{every picture/.style={line width=0.75pt}} 

\begin{tikzpicture}[x=0.75pt,y=0.75pt,yscale=-1,xscale=1]

\draw [line width=2.25]    (220,185) -- (155,185) ;
\draw [shift={(150,185)}, rotate = 360] [fill=myblack  ][line width=0.08]  [draw opacity=0] (14.29,-6.86) -- (0,0) -- (14.29,6.86) -- cycle    ;
\draw [line width=2.25]    (100,285) -- (165,285) ;
\draw [shift={(170,285)}, rotate = 180] [fill=myblack  ][line width=0.08]  [draw opacity=0] (14.29,-6.86) -- (0,0) -- (14.29,6.86) -- cycle    ;
\draw [color=border  ,draw opacity=1 ][line width=2.25]  [dash pattern={on 2.53pt off 3.02pt}]  (135,170) -- (87.91,104.07) ;
\draw [shift={(85,100)}, rotate = 54.46] [fill=border  ,fill opacity=1 ][line width=0.08]  [draw opacity=0] (14.29,-6.86) -- (0,0) -- (14.29,6.86) -- cycle    ;
\draw [line width=2.25]    (185,270) -- (137.91,204.07) ;
\draw [shift={(135,200)}, rotate = 54.46] [fill=myblack  ][line width=0.08]  [draw opacity=0] (14.29,-6.86) -- (0,0) -- (14.29,6.86) -- cycle    ;
\draw [color=border  ,draw opacity=1 ][line width=2.25]  [dash pattern={on 2.53pt off 3.02pt}]  (235,170) -- (187.91,104.07) ;
\draw [shift={(185,100)}, rotate = 54.46] [fill=border  ,fill opacity=1 ][line width=0.08]  [draw opacity=0] (14.29,-6.86) -- (0,0) -- (14.29,6.86) -- cycle    ;
\draw [line width=2.25]    (285,270) -- (237.91,204.07) ;
\draw [shift={(235,200)}, rotate = 54.46] [fill=myblack  ][line width=0.08]  [draw opacity=0] (14.29,-6.86) -- (0,0) -- (14.29,6.86) -- cycle    ;
\draw [line width=2.25]    (185,270) -- (232.09,204.07) ;
\draw [shift={(235,200)}, rotate = 125.54] [fill=myblack  ][line width=0.08]  [draw opacity=0] (14.29,-6.86) -- (0,0) -- (14.29,6.86) -- cycle    ;
\draw [line width=2.25]    (235,170) -- (282.09,104.07) ;
\draw [shift={(285,100)}, rotate = 125.54] [fill=myblack  ][line width=0.08]  [draw opacity=0] (14.29,-6.86) -- (0,0) -- (14.29,6.86) -- cycle    ;
\draw  [color=myblack  ,draw opacity=1 ][fill={rgb, 255:red, 255; green, 255; blue, 255 }  ,fill opacity=1 ][line width=3]  (170,285) .. controls (170,276.72) and (176.72,270) .. (185,270) .. controls (193.28,270) and (200,276.72) .. (200,285) .. controls (200,293.28) and (193.28,300) .. (185,300) .. controls (176.72,300) and (170,293.28) .. (170,285) -- cycle ;
\draw  [color=myblack  ,draw opacity=1 ][fill={rgb, 255:red, 255; green, 255; blue, 255 }  ,fill opacity=1 ][line width=3]  (70,285) .. controls (70,276.72) and (76.72,270) .. (85,270) .. controls (93.28,270) and (100,276.72) .. (100,285) .. controls (100,293.28) and (93.28,300) .. (85,300) .. controls (76.72,300) and (70,293.28) .. (70,285) -- cycle ;
\draw  [color=myblack  ,draw opacity=1 ][fill={rgb, 255:red, 255; green, 255; blue, 255 }  ,fill opacity=1 ][line width=3]  (270,285) .. controls (270,276.72) and (276.72,270) .. (285,270) .. controls (293.28,270) and (300,276.72) .. (300,285) .. controls (300,293.28) and (293.28,300) .. (285,300) .. controls (276.72,300) and (270,293.28) .. (270,285) -- cycle ;
\draw  [color=myblack  ,draw opacity=1 ][fill={rgb, 255:red, 255; green, 255; blue, 255 }  ,fill opacity=1 ][line width=3]  (120,185) .. controls (120,176.72) and (126.72,170) .. (135,170) .. controls (143.28,170) and (150,176.72) .. (150,185) .. controls (150,193.28) and (143.28,200) .. (135,200) .. controls (126.72,200) and (120,193.28) .. (120,185) -- cycle ;
\draw  [color=myblack  ,draw opacity=1 ][fill={rgb, 255:red, 255; green, 255; blue, 255 }  ,fill opacity=1 ][line width=3]  (220,185) .. controls (220,176.72) and (226.72,170) .. (235,170) .. controls (243.28,170) and (250,176.72) .. (250,185) .. controls (250,193.28) and (243.28,200) .. (235,200) .. controls (226.72,200) and (220,193.28) .. (220,185) -- cycle ;
\draw  [color=border  ,draw opacity=1 ][fill={rgb, 255:red, 255; green, 255; blue, 255 }  ,fill opacity=1 ][dash pattern={on 3.38pt off 3.27pt}][line width=3]  (170,85) .. controls (170,76.72) and (176.72,70) .. (185,70) .. controls (193.28,70) and (200,76.72) .. (200,85) .. controls (200,93.28) and (193.28,100) .. (185,100) .. controls (176.72,100) and (170,93.28) .. (170,85) -- cycle ;
\draw  [color=border  ,draw opacity=1 ][fill={rgb, 255:red, 255; green, 255; blue, 255 }  ,fill opacity=1 ][dash pattern={on 3.38pt off 3.27pt}][line width=3]  (70,85) .. controls (70,76.72) and (76.72,70) .. (85,70) .. controls (93.28,70) and (100,76.72) .. (100,85) .. controls (100,93.28) and (93.28,100) .. (85,100) .. controls (76.72,100) and (70,93.28) .. (70,85) -- cycle ;
\draw  [color=myblack  ,draw opacity=1 ][fill={rgb, 255:red, 255; green, 255; blue, 255 }  ,fill opacity=1 ][line width=3]  (270,85) .. controls (270,76.72) and (276.72,70) .. (285,70) .. controls (293.28,70) and (300,76.72) .. (300,85) .. controls (300,93.28) and (293.28,100) .. (285,100) .. controls (276.72,100) and (270,93.28) .. (270,85) -- cycle ;

\draw (81,242.4) node [anchor=north west][inner sep=0.75pt]  [font=\sizeFtwo]  {$\mathbf{t_{0} .x_{0}}$};
\draw (201,252.4) node [anchor=north west][inner sep=0.75pt]  [font=\sizeFtwo]  {$\mathbf{t_{1} .x_{1}}$};
\draw (301,252.4) node [anchor=north west][inner sep=0.75pt]  [font=\sizeFtwo]  {$\mathbf{t_{2} .x_{2}}$};
\draw (141,145.4) node [anchor=north west][inner sep=0.75pt]  [font=\sizeFtwo]  {$\mathbf{t_{4} .x_{4}}$};
\draw (261,172.4) node [anchor=north west][inner sep=0.75pt]  [font=\sizeFtwo]  {$\mathbf{t_{3} .x_{3}}$};
\draw (102,88.4) node [anchor=north west][inner sep=0.75pt]  [font=\sizeFtwo]  {$\mathbf{t_{5} .x_{5}}$};
\draw (202,88.4) node [anchor=north west][inner sep=0.75pt]  [font=\sizeFtwo]  {$\mathbf{t_{6} .x_{6}}$};
\draw (301,95.4) node [anchor=north west][inner sep=0.75pt]  [font=\sizeFtwo]  {$\mathbf{t_{7} .x_{7}}$};
\draw (279,242.4) node [anchor=north west][inner sep=0.75pt]  [font=\sizeFtwop,color={rgb, 255:red, 155; green, 155; blue, 155 }  ,opacity=1 ]  {$:a$};
\draw (201,242.4) node [anchor=north west][inner sep=0.75pt]  [font=\sizeFtwop,color={rgb, 255:red, 155; green, 155; blue, 155 }  ,opacity=1 ]  {$:a$};
\draw (102,288.4) node [anchor=north west][inner sep=0.75pt]  [font=\sizeFtwop,color={rgb, 255:red, 155; green, 155; blue, 155 }  ,opacity=1 ]  {$:a$};
\draw (99,152.4) node [anchor=north west][inner sep=0.75pt]  [font=\sizeFtwop,color={rgb, 255:red, 155; green, 155; blue, 155 }  ,opacity=1 ]  {$:a$};
\draw (69,112.4) node [anchor=north west][inner sep=0.75pt]  [font=\sizeFtwop,color={rgb, 255:red, 155; green, 155; blue, 155 }  ,opacity=1 ]  {$:a$};
\draw (169,112.4) node [anchor=north west][inner sep=0.75pt]  [font=\sizeFtwop,color={rgb, 255:red, 155; green, 155; blue, 155 }  ,opacity=1 ]  {$:a$};
\draw (279,112.4) node [anchor=north west][inner sep=0.75pt]  [font=\sizeFtwop,color={rgb, 255:red, 155; green, 155; blue, 155 }  ,opacity=1 ]  {$:a$};
\draw (249,152.4) node [anchor=north west][inner sep=0.75pt]  [font=\sizeFtwop,color={rgb, 255:red, 155; green, 155; blue, 155 }  ,opacity=1 ]  {$:a$};
\draw (251,202.4) node [anchor=north west][inner sep=0.75pt]  [font=\sizeFtwop,color={rgb, 255:red, 155; green, 155; blue, 155 }  ,opacity=1 ]  {$:b$};
\draw (199,202.4) node [anchor=north west][inner sep=0.75pt]  [font=\sizeFtwop,color={rgb, 255:red, 155; green, 155; blue, 155 }  ,opacity=1 ]  {$:c$};
\draw (137,288.4) node [anchor=north west][inner sep=0.75pt]  [font=\sizeFtwop,color={rgb, 255:red, 155; green, 155; blue, 155 }  ,opacity=1 ]  {$:b$};
\draw (151,202.4) node [anchor=north west][inner sep=0.75pt]  [font=\sizeFtwop,color={rgb, 255:red, 155; green, 155; blue, 155 }  ,opacity=1 ]  {$:b$};
\draw (161,162.4) node [anchor=north west][inner sep=0.75pt]  [font=\sizeFtwop,color={rgb, 255:red, 155; green, 155; blue, 155 }  ,opacity=1 ]  {$:c$};
\draw (198,152.4) node [anchor=north west][inner sep=0.75pt]  [font=\sizeFtwop,color={rgb, 255:red, 155; green, 155; blue, 155 }  ,opacity=1 ]  {$:d$};
\draw (151,252.4) node [anchor=north west][inner sep=0.75pt]  [font=\sizeFtwop,color={rgb, 255:red, 155; green, 155; blue, 155 }  ,opacity=1 ]  {$:c$};

\end{tikzpicture}}
    \caption{$G$}
    \label{fig : Induced subgraph 1}
\end{subfigure}
\hfil
\hfil
\begin{subfigure}{0.41\textwidth}
    \captionsetup{justification=centering}
    \centering
    \resizebox{.5\textwidth}{!}{\tikzset{every picture/.style={line width=0.75pt}} 

\begin{tikzpicture}[x=0.75pt,y=0.75pt,yscale=-1,xscale=1]

\draw [color=border  ,draw opacity=1 ][line width=2.25]  [dash pattern={on 2.53pt off 3.02pt}]  (222,185) -- (157,185) ;
\draw [shift={(152,185)}, rotate = 360] [fill=border  ,fill opacity=1 ][line width=0.08]  [draw opacity=0] (14.29,-6.86) -- (0,0) -- (14.29,6.86) -- cycle    ;
\draw [color=border  ,draw opacity=1 ][line width=2.25]  [dash pattern={on 2.53pt off 3.02pt}]  (237,170) -- (189.91,104.07) ;
\draw [shift={(187,100)}, rotate = 54.46] [fill=border  ,fill opacity=1 ][line width=0.08]  [draw opacity=0] (14.29,-6.86) -- (0,0) -- (14.29,6.86) -- cycle    ;
\draw [line width=2.25]    (287,270) -- (239.91,204.07) ;
\draw [shift={(237,200)}, rotate = 54.46] [fill=myblack  ][line width=0.08]  [draw opacity=0] (14.29,-6.86) -- (0,0) -- (14.29,6.86) -- cycle    ;
\draw [color=border  ,draw opacity=1 ][line width=2.25]  [dash pattern={on 2.53pt off 3.02pt}]  (187,270) -- (234.09,204.07) ;
\draw [shift={(237,200)}, rotate = 125.54] [fill=border  ,fill opacity=1 ][line width=0.08]  [draw opacity=0] (14.29,-6.86) -- (0,0) -- (14.29,6.86) -- cycle    ;
\draw [line width=2.25]    (237,170) -- (284.09,104.07) ;
\draw [shift={(287,100)}, rotate = 125.54] [fill=myblack  ][line width=0.08]  [draw opacity=0] (14.29,-6.86) -- (0,0) -- (14.29,6.86) -- cycle    ;
\draw  [color=border  ,draw opacity=1 ][fill={rgb, 255:red, 255; green, 255; blue, 255 }  ,fill opacity=1 ][dash pattern={on 3.38pt off 3.27pt}][line width=3]  (172,285) .. controls (172,276.72) and (178.72,270) .. (187,270) .. controls (195.28,270) and (202,276.72) .. (202,285) .. controls (202,293.28) and (195.28,300) .. (187,300) .. controls (178.72,300) and (172,293.28) .. (172,285) -- cycle ;
\draw  [color=myblack  ,draw opacity=1 ][fill={rgb, 255:red, 255; green, 255; blue, 255 }  ,fill opacity=1 ][line width=3]  (272,285) .. controls (272,276.72) and (278.72,270) .. (287,270) .. controls (295.28,270) and (302,276.72) .. (302,285) .. controls (302,293.28) and (295.28,300) .. (287,300) .. controls (278.72,300) and (272,293.28) .. (272,285) -- cycle ;
\draw  [color=border  ,draw opacity=1 ][fill={rgb, 255:red, 255; green, 255; blue, 255 }  ,fill opacity=1 ][dash pattern={on 3.38pt off 3.27pt}][line width=3]  (122,185) .. controls (122,176.72) and (128.72,170) .. (137,170) .. controls (145.28,170) and (152,176.72) .. (152,185) .. controls (152,193.28) and (145.28,200) .. (137,200) .. controls (128.72,200) and (122,193.28) .. (122,185) -- cycle ;
\draw  [color=myblack  ,draw opacity=1 ][fill={rgb, 255:red, 255; green, 255; blue, 255 }  ,fill opacity=1 ][line width=3]  (222,185) .. controls (222,176.72) and (228.72,170) .. (237,170) .. controls (245.28,170) and (252,176.72) .. (252,185) .. controls (252,193.28) and (245.28,200) .. (237,200) .. controls (228.72,200) and (222,193.28) .. (222,185) -- cycle ;
\draw  [color=border  ,draw opacity=1 ][fill={rgb, 255:red, 255; green, 255; blue, 255 }  ,fill opacity=1 ][dash pattern={on 3.38pt off 3.27pt}][line width=3]  (172,85) .. controls (172,76.72) and (178.72,70) .. (187,70) .. controls (195.28,70) and (202,76.72) .. (202,85) .. controls (202,93.28) and (195.28,100) .. (187,100) .. controls (178.72,100) and (172,93.28) .. (172,85) -- cycle ;
\draw  [color=myblack  ,draw opacity=1 ][fill={rgb, 255:red, 255; green, 255; blue, 255 }  ,fill opacity=1 ][line width=3]  (272,85) .. controls (272,76.72) and (278.72,70) .. (287,70) .. controls (295.28,70) and (302,76.72) .. (302,85) .. controls (302,93.28) and (295.28,100) .. (287,100) .. controls (278.72,100) and (272,93.28) .. (272,85) -- cycle ;

\draw (203,252.4) node [anchor=north west][inner sep=0.75pt]  [font=\sizeFtwo]  {$\mathbf{t_{1} .x_{1}}$};
\draw (303,252.4) node [anchor=north west][inner sep=0.75pt]  [font=\sizeFtwo]  {$\mathbf{t_{2} .x_{2}}$};
\draw (141,142.4) node [anchor=north west][inner sep=0.75pt]  [font=\sizeFtwo]  {$\mathbf{t_{4} .x_{4}}$};
\draw (263,172.4) node [anchor=north west][inner sep=0.75pt]  [font=\sizeFtwo]  {$\mathbf{t_{3} .x_{3}}$};
\draw (204,88.4) node [anchor=north west][inner sep=0.75pt]  [font=\sizeFtwo]  {$\mathbf{t_{6} .x_{6}}$};
\draw (303,95.4) node [anchor=north west][inner sep=0.75pt]  [font=\sizeFtwo]  {$\mathbf{t_{7} .x_{7}}$};
\draw (281,242.4) node [anchor=north west][inner sep=0.75pt]  [font=\sizeFtwop,color={rgb, 255:red, 155; green, 155; blue, 155 }  ,opacity=1 ]  {$:a$};
\draw (203,242.4) node [anchor=north west][inner sep=0.75pt]  [font=\sizeFtwop,color={rgb, 255:red, 155; green, 155; blue, 155 }  ,opacity=1 ]  {$:a$};
\draw (171,112.4) node [anchor=north west][inner sep=0.75pt]  [font=\sizeFtwop,color={rgb, 255:red, 155; green, 155; blue, 155 }  ,opacity=1 ]  {$:a$};
\draw (281,112.4) node [anchor=north west][inner sep=0.75pt]  [font=\sizeFtwop,color={rgb, 255:red, 155; green, 155; blue, 155 }  ,opacity=1 ]  {$:a$};
\draw (251,152.4) node [anchor=north west][inner sep=0.75pt]  [font=\sizeFtwop,color={rgb, 255:red, 155; green, 155; blue, 155 }  ,opacity=1 ]  {$:a$};
\draw (253,202.4) node [anchor=north west][inner sep=0.75pt]  [font=\sizeFtwop,color={rgb, 255:red, 155; green, 155; blue, 155 }  ,opacity=1 ]  {$:b$};
\draw (201,202.4) node [anchor=north west][inner sep=0.75pt]  [font=\sizeFtwop,color={rgb, 255:red, 155; green, 155; blue, 155 }  ,opacity=1 ]  {$:c$};
\draw (163,162.4) node [anchor=north west][inner sep=0.75pt]  [font=\sizeFtwop,color={rgb, 255:red, 155; green, 155; blue, 155 }  ,opacity=1 ]  {$:c$};
\draw (200,152.4) node [anchor=north west][inner sep=0.75pt]  [font=\sizeFtwop,color={rgb, 255:red, 155; green, 155; blue, 155 }  ,opacity=1 ]  {$:d$};

\end{tikzpicture}}
    \caption{$G_{\{x_2,x_3,x_7\}}$}
    \label{fig : Induced subgraph 2}
\end{subfigure}
\caption{
{\em Induced subgraphs and borders.} 
$(a)$ A graph $G$ and $(b)$ its induced subgraph $G_{\{x_2,x_3,x_7\}}$. Both graphs have borders, as shown by the dashed lines. 
}
\label{fig : Induced subgraph}
\end{figure}

Next, we define the induced subgraph $G_U \sqsubseteq G$ as the graph whose internal vertices are $I_{G_U}=I_G\cap U$, and whose edges are all those edges of $G$ which touch a vertex in $I_{G_U}$. Thus, its border vertices $B_{G_U}$ are those nodes of $V_G\setminus I_{G_U}$ which lie at distance one of $I_{G_U}$ in $G$ (see Fig.~\ref{fig : Induced subgraph}). We also introduce the operation $G\sqcup H$, a union which is only defined if both $G$ and $H$ can be viewed as induced subgraphs of the same larger graph. In particular, this implies that if $u\in I_G$,  $v\in \B_G$, $v\in I_H$ and $(u\!:\!a,v\!:\!b)\in E_G$, it must be the case that $(u\!:\!a,v\!:\!b)\in E_H$ and $u\in V_H$. This union will allow us to express locality. In practice it is often convenient to restrict to a subset of the set of all graphs, but then we need to assume a number of closure properties: 
\begin{definition}[Closed subset of graphs]\label{def : closure}
Consider $\sshift\subseteq\mathcal{G}$. This $\sshift$ is said to be closed under
\begin{description}
    \item{disjoint-union:} $G, H\in \sshift$ and $V_G\cap V_H=\varnothing$ implies $G\sqcup H\in \sshift$.
    \item{renaming:} $G\in \sshift$ implies $R G\in \sshift$.
    \item{\maximality:} \MC{ $G\in \sshift$ implies that $\exists G'\in \sshift\maxi$ such that $G\sqsubseteq G'$.}
    \item{\bckwmaximality:} \MC{ $G\in \sshift$ implies that  $\exists H\in {}\maxi\!\mathcal{S}$ such that $G\sqsubseteq G'$.}
    \item{inclusion:} $G\in \sshift$ implies $\forall H\sqsubseteq G,\ H\in \sshift$.
\end{description}    
\end{definition}
In the following theoretical results, $\sshift$ is any subset of $G$ respecting these four closure properties.

A cone of $x$ is the kind of subgraphs obtained by exploring a graph by starting from a node at position $x$ and then moving along the directed edges. When we take the neighbourhood of a graph at position $x$, we look for a cone at that position:
\begin{definition}[Cones and neighbourhood scheme]\label{def:neighbourhood}
Consider $u\in \mathcal{V}$. We denote by $\upast{u}=\{G\in\sshift\,|\, u\in \Past(G)\}$ the set of graphs for which $u$ is past.\\
A \emph{(forward) cone} of $x$ is a graph $C\in\sshift$ such that for every $v\in I_C$ there exists a directed path from the vertex at position $x$ to $v$ that lies entirely in $I_C$, a.k.a {\em accessibility}. We denote by $\xCone$ the set of cones of $x$, and by $\xcone{}$ the set of all cones.\\

A \emph{(forward) neighbourhood scheme} $\restri{}$ is a renaming-invariant function which maps any position $x\in \mathcal{X}$ to a function $\restri{x} : \graphvalidexplicit{\restri{}}{x} \to P(\mathcal{X})$ on some domain $\graphvalidexplicit{\restri{}}{x} \subseteq \sshift$ such that for any $G, H \in \sshift$ we have:
\begin{description}
    \item{completeness:} 
    $\forall G\in \upast{u}\maxi$ \RX{s.t. $x=\sp{u}$,}{, $\exists x\in\X,\ x\namesubseteq\sp{u}$,}\ $G\in\graphvalid{\restri{}}{x}$.
    \RX{}{\item{unambiguity:} $\forall G\in \upast{u},\ |\{x\in\X \mid x\namesubseteq\sp{u}\textrm{ and }G\in\graphvalid{\restri{}}{x}\}|\leq 1$.}
    \item{cone:} $G\in \graphvalid{\restri{}}{x}$ implies $\restr{x}{G}\subseteq\sp{I_G}$ and $G_{\restri{x}}$ is in $\xCone$,
    \item{strong extensivity:} $G\in \graphvalid{\restri{x}}{}$ and $G_{\restri{x}} \sqsubseteq H \text{ implies } H\in\graphvalid{\restri{}}{x}, G_{\restri{x}} = H_{\restri{x}}.$
\end{description}
where we introduced $G_{\restri{x}}$ as a shorthand for $G_{\restr{x}{G}}$.\\
We denote $\diskA{x}= \{ G_{\restri{x}}\mid G\in \sshift \}$ the set of disks of $x$.\\
A {\em backwards neighbourhood scheme} $\restriB{}$ is defined by the application of a forward neighbourhood scheme $\restriA{}$ on $\transpose{G}$, i.e. $\restrB{x}{G}:=\restr{x}{\transpose{G}}$. 

\end{definition}
\noindent Notice how we let some flexibility on the domain of definition $\graphvalid{\restri{}}{x}$ of $\restri{x}$. For the \maxgraph graphs, there has to exist some position $x$ of a past node $u$ so that $\restri{x}$ be defined. 
But this does not need to happen when $u$ is close to a dangling outgoing edge.
However, once such a neighbourhood is defined, strong extensivity ensures that this is the case for all the graphs that contain it.

Next, a local operator $A_x$ is one that reads/writes just the neighbourhood $\restr{x}{G}$. That is unless the neighbourhood is undefined, in which case it `waits' i.e. acts as the identity.

\begin{definition}[$\restri{}$-local operator]\label{def : locality}
    Given some forward neighbourhood scheme $\restri{}$, an \emph{$\restri{}$-local operator} $A_{(-)}$ is an renaming invariant operator from $\mathcal{X}\times\sshift$ to $\sshift$ such that $\forall x\in\X$, 
    $\forall G \in \sshift,$
    $$A_x G := 
    \begin{cases}
    (A_x G_\restri{x}) \sqcup G_\crestri{x}&\textrm{if $G\in \graphvalid{\restri{}}{x}$}\\
    G&\textrm{otherwise}
    \end{cases}
    $$
\end{definition}

\subsection{A more constructive characterisation}

While the above definitions of neighbourhood schemes and local operators express directly what properties are expected, locality allows for a more concise characterisation of these objects.
Indeed locality tells us that it is enough to concentrate on graphs that are entirely accessible from past a past node at position $x$. The part of a \maxgraph graph that is accessible from $x$ is the \maxgraph cone of $x$, it contains every possible neighbourhood of $x$.
With this in mind, let us begin by giving an alternative presentation of a neighbourhood schemes, namely as a mutex set of cones of $x$, and make explicit the bijection that relates the two. Pushing this line of thoughts will allow us to establish a connection with rewriting theory. 

%

\begin{definition}[Mutex]\label{def : mutex cones}
%
    The renaming-invariant map $\mathcal{D}:\X\to\calP(\xcone{}),\ x\mapsto \mathcal{D}_x$ defines \emph{a mutex set of cones} iff it is 
    \begin{description} 
    \item{well-indexed:} $\mathcal{D}_x\subseteq\xcone{x}$\RX{}{ and $\mathcal{D}_x\cap \mathcal{D}_y\neq \varnothing$ implies $x=y$}.
    \item{complete:} $\forall C\in \xcone{}\maxi$, $\exists x\in\X,\exists D\in \mathcal{D}_x, D\sqsubseteq C$.
    \item{unambiguous:} $\forall C \in \xcone{},$
    $|\{ D\in\xcone{} \mid \exists x, D\in \mathcal{D}_x\textrm{ and }D\sqsubseteq C\}|\leq 1$.
    \end{description}
    We call $\graphvalidmutex{\D}{x}=\{G\in \sshift\mid \exists D\in\D_x , D\sqsubseteq G\}$ the set of graphs containing some $D\in\mathcal{D}_x$. Given a graph $G\in \graphvalidmutex{\D}{x}$, unambiguity ensures that this $D$ is unique. This allows us to define $\restriD{\mathcal{D}}{x}:\graphvalidmutex{\D}{x}\to P(\X),\,G\mapsto \restrD{\mathcal{D}}{x}{G}=\sp{\I_{D}}$. Notice that $G_{\restriD{\mathcal{D}}{x}}=D$. 
\end{definition}
\begin{restatable}[Neighbourhood schemes as mutex sets of cones]{lemma}{lemmamutexneighbourhood}\label{lemma : equivalence set mutex and neighbourhood schemes}
    For any neighbourhood scheme $\restriA{}$ and $x\in\X$, let $\diskA{x} := \{G_{\restriA{x}} \mid G\in \sshift \}$. Then $\diskA{}:x\mapsto \diskA{x}$ defines a mutex set of cones. This correspondence is a bijection between neighbourhood schemes on the one hand, and mutex sets of cones on the other. The inverse associates, to any mutex set of cones $\D$, the neighbourhood scheme $\restriD{\D}{}$ that maps $x$ to $\restriD{\D}{x}$. 
\end{restatable}

Let us now show that local operators $A_x$ act as graph rewrite systems, which rewrite disks in a deterministic manner.

\begin{definition}[Causal rewrite system]\label{def : deterministic rewrite system}
    A set of rewrite rules $\{D_x^j\to G_x^j\}_{x\in\X,j\in J}$ specifies a (forward) {\em causal rewrite system} over $\sshift$ iff it is 
    \begin{description}
        \item{functional:} $D_x^j=D_y^k$ implies $(x,j)=(y,k)$, so that $f:D_x^j\mapsto G_x^j$ is a function.
        \item{renaming-invariant:} for the function $f$. 
        \item{mutex-domain:} ${\cal D}:x\mapsto\{D_x^j \mid j\in J \}$ defines a mutex set of cones.
        \item{context-preserving:} $\I_{G_x^j}\namesubseteq \I_{D_x^j}$ and $F_{G_x^j} = F_{D_x^j}$, with $F_G = \E_G \setminus (\I_G:\pi)^2$.
        \item{$\sshift$-preserving:} $D_x^j\sqcup H\in \sshift \implies G_x^j\sqcup H \in \sshift$
\end{description}
A set of rewrite rules $\{{D_x^j}\to {G_x^j}\}_{x\in\X,j\in J}$ is a {\em backwards} causal rewrite system over $\sshift$ if its transposition $\{\transpose{D_x^j}\to \transpose{G_x^j}\}_{x\in\X,j\in J}$ is a causal rewrite system over $\transpose{\sshift}$.
\end{definition}

\begin{restatable}[Local operators as a causal rewrite systems]{proposition}{propcharacmutex}\label{lemma : charac mutex}
    For any pair formed by a neighbourhood scheme $\restriA{}$ and a $\restriA{}$-local rule $A_{(-)}$, the family of rules $\{ C \to A_x C\}_{C \in \diskA{x},x\in\X}$ is a causal rewrite system.
    Moreover, this construction is a bijection 
    whose inverse associates to any causal rewrite system $\{D_x^j \to G_x^j\}_{x\in\X,j\in J}$, the pair formed by a neighbourhood scheme $\restriD{\D}{}$ with ${\cal D}:x\mapsto\{D_x^j \mid j\in J\}$, and the $\restriD{\D}{}$-local operator $A_{(-)}$ defined by
        $$A_x G := H \text{ if and only if } G\to_x H,
    $$
where rewriting $G$ at $x$ means replacing one or no occurrence of $D_x^j\in \D_x$ in $G$ by $G_x^j$ to yield $H:=G[G_x^j/D_x^j]$.
\end{restatable}

This first proposition allows us to establish a formal correspondence between ``applying the local operators of space-time deterministic graph rewriting'', and ``rewrite systems'' in the more traditional sense of input/output pattern replacement. But it also provides us with a way of constructing local operators.  

\section{Reversibility} \label{sec:rev}

In cellular automata theory, reversibility refers not only to the global function having an inverse, but also that the inverse be itself a cellular automata~\cite{KariBlock}. Here we demand that a local operator not only has an inverse, but also that the inverse be a local operator:   
\begin{definition}[Space-time reversible]\label{def:rev}
Given some forward neighborhood scheme $\restriA{}$ and an $\restriA{}$-local operator $A_{( -)}$, we say that a $A_{(-)}$ is \emph{(space-time) reversible} iff there exists a backwards neighbourhood scheme $\restriB{}$ and an $\restriB{}$-local operator $B_{( -)}$ such that

\begin{description}
    \item{left local inverse:} $\forall G\in\graphvalidexplicit{\restri{}}{x},\, B_x A_x G = G \textrm{ with }\restrB{x}{A_x G}\nameeq\restrA{x}{G}.$
    \item{right local inverse:} $\forall H\in\graphvalidexplicit{\restriB{}}{x},\, A_x B_x H = H
    \textrm{ with }\restrA{x}{B_xH}\nameeq\restrB{x}{H}.$
\end{description}
\end{definition}


Interestingly, for a reversible local operator, renaming-invariance implies name-preservation, i.e $\hat{V}_{A_x G}=\hat{V}_{G}$, as we have proven in Prop.~\ref{prop:renaminginvarianceimpliesnp}, which we have pushed to App.~\ref{sec:name-preservation}.


Notice that a reversible rule must always transform the positions of past vertices into that of future vertices. Indeed, $G\in \graphvalid{\restriA{}}{x}$ is equivalent to $A_x G\in \graphvalid{\restriB{}}{x}$, and so \RX{$x=\sp{u}$}{$x\namesubseteq\sp{u}$} with $u\in \Past(G)$ is equivalent to \RX{$x=\sp{u}$ with $u\in \Fut(G)$}{$x\namesubseteq\sp{v}$ with $v\in \Fut(G)$}.\\
As a first sanity check we have:
\begin{restatable}[Uniqueness of the inverse]{lemma}{lemuniqueness}
    Let $A_{(-)}$ be a local operator. Let $B_{(-)}$ and $B'_{(-)}$ be two inverses of $A_{(-)}$. Then $B_{(-)}=B'_{(-)}$.
\end{restatable}
Whilst Def.~\ref{def:rev} is certainly the natural one, it demands many properties at once, cf. Def.~\ref{def:neighbourhood} and~\ref{def : locality}, both for $A_x$ and $B_x$. This may be cumbersome to check, and makes us wonder about the `gap' between just asking that $A_x$ be invertible. The following proposition settles this question:
\begin{restatable}[Axiomatic characterisation of reversibility]{proposition}{propcharacterizationreversibility}\label{prop:inv-rev}
Given some forward neighborhood scheme $\restriA{}$ and an $\restriA{}$-local operator $A_{(-)}$, we say that $A_{(-)}$ is \emph{axiomatic-reversible} iff
\begin{description}
    \item{injectivity:} $A_x$ is injective on $\graphvalidexplicit{\restriA{}}{x}$.
    \item{surjectivity:} 
    $\forall H\in{}\maxi\!\ufut{u}$ \RX{ s.t. $x=\sp{u}$}{,$\exists x\in\X, x\namesubseteq\sp{u}$},\ $\exists G\in\graphvalidexplicit{\restriA{}}{x} \textrm{ s.t. } H=A_xG$. 
    \item{back-reachability:} $\forall D\in \diskA{x}$,
    \begin{itemize}
        \item $A_x D$ is a backwards cone of $x$ 
        \item $(A_x D\sqcup K)\in\sshift$ implies $(D\sqcup K)\in\sshift$.
    \end{itemize}
\end{description}
where $\ufut{u} = \{ G\mid G\in \sshift \textrm{ and } u\in \Fut(G)\}$.\\
Axiomatic-reversibility is equivalent to reversibility.   
\end{restatable}
Thus, reversibility can be expressed quite mathematically in terms of these few natural conditions imposed upon the forward local operator $A_x$. In particular this means that its inverse automatically inherits all of the properties that make it a local operator $B_x$ (completeness, strong extensivity, renaming-invariance\ldots). This is what we have had to prove in order to reach this result, see App.~\ref{app:reversibility}.


Still, at this stage, it is not so clear how to come up with such and $A_x$, let alone be exhaustively enumerating them. This is why we came up with a, third, constructive characterization of reversible local operators:
\begin{restatable}[Reversibility as a two-way causal rewrite system]{proposition}{proprevastwowaycasaulrewsys}\label{prop : mutex reversibility}
    Let $\{D_x^j\to E_x^j\}_{x\in\X,j\in J}$ be the causal rewrite system over $\sshift$ characterizing a local operator $A_{(-)}$. 
    $A_{(-)}$ is reversible if and only $\{{E_x^j}\to {D_x^j}\}_{x\in\X,j\in J}$ is a backwards causal rewrite system in ${S}$.
\end{restatable}

The natural, axiomatic and constructive definitions of reversibility complement each other for different purposes. Their equivalence suggests that we have reached a quite robust notion.

\section{Space-time reversibility}\label{sec:commutation}


\noindent {\em Valid sequences.} From now on $A_{yx}$ will stand for $A_y A_x$. We say that $\omega \in \mathcal{X}^*$ is a valid sequence in $G$ if, for all $\omega_1,\omega_2\in \mathcal{X}^*$ such that $\omega = \omega_2 x \omega_1$, we have $A_{\omega_1}G\in \graphvalid{\restriA{}}{x}$. We denote $\Omega_G(A) \subseteq \mathcal{X}^*$ the set of valid sequences in $G$.\\

\noindent {\em Space-time determinism.} The notion of space-time determinism of a local rule $A_{(-)}$ was developed in \cite{ArrighiDetRew}. Let us remind that, intuitively, it is the idea that the non-determinism of scheduling (i.e. given any two valid sequences $\omega_1$ and $\omega_2$, which of $A_{\omega_1}$ or $A_{\omega_2}$ gets applied to $G$) does not matter as far as space-time is concerned (if is so happens that node $u$ appears in both, then $A_{\omega_1}G$ and $A_{\omega_2}G$ must be consistent about its state and connectivity).\\
Actually, the notion consistency is forcibly more subtle, as discussed at length in \cite{ArrighiDetRew}. At best we can require `full consistency', namely that if $A_{\omega_1}G$ and $A_{\omega_2}G$ agree upon the set of incoming ports of $u$, then they fully agree on $u$. Often however one may look just for `weak consistency', namely that if $A_{\omega_1}G$ and $A_{\omega_2}G$ agree that $u$ has no incoming ports, then they fully agree on $u$. Indeed, having no incoming ports corresponds to the idea that the process at $u$ has no dependency and thus holds a ``result state''.\\
In order to obtain the weak consistency of a local rule $A_{(-)}$, all we need is:
\begin{definition}[Time-increasing commutative local rules] A local rule $A_{(-)}$ is
\begin{itemize}
    \item \emph{time-increasing} 
    iff $\forall G\in\graphvalidexplicit{\restriA{x}},\,\forall t.y \in V_G,\,\forall t'.y'\in V_{A_x G},\ y\namecap y'\neq\varnothing \textrm{ implies }t \leq t'\text{, with }t<t'$ if \RX{$x=y$}{$x\namesubseteq y$};
    \item \emph{commutative} iff $\forall G\in \upast{u}\cap\upast{v}$
    \RX{ s.t. $x=\sp{u}, y=\sp{v}$}{,$\forall x,y\in\X$ with $x\namesubseteq \sp{u}, y\namesubseteq \sp{v}$}, 
    we have 
    $A_{xy} G = A_{yx} G$.
\end{itemize}
\end{definition}
In order to obtain full consistency, a couple more properties are required, which we do not need to remind here.

\noindent {\em Consistency of the inverse.} For a reversible commutative local rule $A_{(-)}$ with local inverse $B_{(-)}$, the question naturally arises whether $B_{(-)}$ inherits the full/weak consistency of $A_{(-)}$. The following result indeed reduces the full/weak consistency of $B_{(-)}$ to that of $A_{(-)}$, so long as $B_{(-)}$ is also commutative:

\begin{restatable}{proposition}{inversedet}\label{prop : the inverse dynamics has the same property}
    Let $A_{(-)}$ be a commutative reversible local rule whose inverse is also a commutative local rule $B_{(-)}$. For all graph $G\in \sshift$ and sequences $\omega_1,\omega_2\in \valid{B}{G}$, there exists a graph $H\in \sshift$ and some sequences $\omega_1',\omega_2'\in \valid{A}{H}$ such that 
    $$\B_{\omega_1} G = A_{\omega_1'} H\quad\textrm{and}\quad \B_{\omega_2} G = A_{\omega_2'} H.$$
\end{restatable}

We may also wonder what happens if we blend applications of $A_{(-)}$ and applications of $B_{(-)}$. A natural way of doing so is to define a local rule $C_{(-)}$ which behaves like $A_x$ or $B_x$ depending whether the $x$ position is that that of a past or future node:

\begin{definition}[Two-way local rule]
    Consider $A_{(-)}$ a reversible $\restriA{}$-local operator and $B_{(-)}$ its $\restriB{}$-local inverse. 
    We define the {\em two-way} local operator $C_{(-)}$ from $\X \times \sshift$ to $\sshift$ such that $\forall x\in\X$, $\forall G \in \sshift$,
    $$A_x G := 
    \begin{cases}
    A_x G&\textrm{if $G\in \graphvalid{\restriA{}}{x}^*$}\\
    B_x G&\textrm{if $G\in \graphvalid{\restriB{}}{x}^*$}\\
    G&\textrm{otherwise}.
    \end{cases}
    $$
    where $\graphvalid{\restriA{}}{x}^* = \graphvalid{\restriA{}}{x} \setminus {\cal O}$ with ${\cal O}$ the set of all graphs containing an isolated vertex, that is a vertex linked to no edges.
\end{definition}

This time $\Omega_G(C)$ the set of valid sequences in $G$ is defined so that $\omega = \omega_2 x \omega_1\in $, we have $C_{\omega_1}G\in \graphvalid{\restriA{}}{x}^*\cup \graphvalid{\restriB{}}{x}^*$ . 
The following two results again reduce the full/weak consistency of $C_{(-)}$ to that of $A_{-}$, so long as $B_{(-)}$ is also commutative: 
\begin{restatable}[Two-way commutation]{lemma}{lemtwowaycommutation}\label{lemma : two-way commutation}
\todo[inline]{Could we modify it to make it look like the commutation hypothesis? It seems not true even for simple examples.}
    Consider $A_{(-)}$ a commutative reversible $\restriA{}$-local operator which admits a commutative $\restriB{}$-local inverse $B_{(-)}$. 
    Let $C_{(-)}$ be the corresponding {\em two-way} local operator.
    Then for all graph $G\in \graphvalid{\restriB{x}}{}^*$ such that $B_xG\in\graphvalid{\restriA{y}}{}^*$ we have $G\in\graphvalid{\restriA{y}}{}^*$, $A_y G\in \graphvalid{\restriB{x}}{}^*$ and :
    $$  C_y C_x G = A_y B_x G = B_x A_y G = C_x C_y G$$
    This also stands for $G\in \graphvalidexplicit{\restriA{x}}{}^*$ and $B_x G\in \graphvalid{\restriB{y}}{}^*$---i.e. then we have $G\in \graphvalid{\restriB{y}}{}^*$, $A_x G\in \graphvalid{\restriA{y}}{}^*$ and $C_yC_x G = C_xC_y G$.
\end{restatable}

\begin{restatable}[Two-way consistency]{proposition}{proptwowayconsistency}
    Consider $A_{(-)}$ a reversible $\restriA{}$-local operator and $B_{(-)}$ its $\restriB{}$-local inverse. 
    Let $C_{(-)}$ be the corresponding {\em two-way} local operator.
    For all graph $G\in\sshift$ and $\omega_1,\omega_2\in \Omega_G(C)$, there exists \PA{$G'\in\sshift$} and $\omega_1',\omega_2'\in \Omega_{G'}(A)$ such that :
    $$C_{\omega_1} G = A_{\omega_1'} G'$$
    $$C_{\omega_2} G = A_{\omega_2'} G'$$
\end{restatable}

Summarizing,
\begin{corollary}
    Consider $A_{(-)}$ a reversible $\restriA{}$-local operator which admits a commutative inverse and $C_{(-)}$ its corresponding {\em two-way} local operator. If $A_{(-)}$ is fully/weakly consistent then so is $C_{(-)}$.
\end{corollary}

\noindent {\em Commutativity of the inverse.} We have seen that the good behaviour of the inverse local rule $B_{(-)}$ or the two-way local rule $C_{(-)}$ both depend on the commutativity of $B_{(-)}$. Thus the new question that arises is whether $B_{(-)}$ inherits the commutativity of $A_{(-)}$. We have identified three important scenarios for which this is the case.\\
First, if $A_{(-)}$ is time-symmetric \cite{GajardoTimeSym,ArrighiBLOCKREP}:
\begin{restatable}[Time-symmetric commutativity]{lemma}{lemtimesymcommutation}
    Let $A_{(-)}$ be a reversible and commutative
    $\restriA{}$-local rule whose inverse local rule is
    $$B_{(-)}=\transpose{~}\circ A^-_{(-)} \circ \transpose{~}.$$
    where $A^-_{(-)}$ acts as $A_{(-)}$ except it modifies time tags in the opposite way. Then $B_{(-)}$ is commutative.
\end{restatable}
Second, if $A_{(-)}$ is such that two future nodes always come from two past nodes: 
\begin{restatable}[Two-two implies commutative inverse]{proposition}{proptwotwo}
\label{lemma : reducing hypothesis for commutation of B}
Let $A_{(-)}$ be a reversible local rule.
We say that it is {\em two-two} iff for all $H\in\graphvalidexplicit{\restriB{}}{x}\cap\graphvalidexplicit{\restriB{}}{y}$ 
there exists $G\in\graphvalidexplicit{\restriA{}}{x}\cap\graphvalidexplicit{\restriA{}}{y}$ such that $H=A_{xy}G$.\\
Conversely, we say that its inverse $B_{(-)}$ is {\em two-two} iff for all $G\in\graphvalidexplicit{\restriA{}}{x}\cap\graphvalidexplicit{\restriA{}}{y}$ there exists $H\in\graphvalidexplicit{\restriB{}}{x}\cap\graphvalidexplicit{\restriB{}}{y}$ such that $G=B_{xy}H$.\\
We have that if $A_{(-)}$ is commutative, then $B_{(-)}$ is two-two, and symmetrically.\\  
We have that $A_{(-)}$ is commutative and two-two if and only if $B_{(-)}$ also is.
\end{restatable}
Third, if $A_{(-)}$ is bounded:
\begin{restatable}[Bounded neighbourhood implies commutative inverse]{lemma}{lemcommutation}
    Let $A_{(-)}$ be a reversible and commutative 
    $\restriA{}$-local rule
    . Let $B_{(-)}$ be the $\restriB{}$-local inverse of $A_{(-)}$. If $\restriA{}$ is bounded by $k$ and there exists a renaming invariant bijection between $\graphvalid{\restriA{}}{x}\cap\graphvalid{\restriA{}}{y}$ and $ \graphvalid{\restriB{}}{x}\cap \graphvalid{\restriB{}}{y}$, then $B_{(-)}$ is commutative.
\end{restatable}




\section{Reversible time-dilation}\label{sec:timedilation}


In~\cite{ArrighiDetRew}, it has be shown how asynchronous applications of a local operator help express phenomena that go beyond the mere asynchronous simulation of a synchronous dynamical system. Namely, a local rule was designed to represent particles moving left and right, but on a background that is subject to `time dilation' in analogy with relativistic physics. In this Section we show that this is still doable in the presence of reversibility.
\begin{figure}[t]
\centering
\begin{minipage}{0.39\textwidth}
    \begin{subfigure}{\textwidth}
    \centering
    \includegraphics[width=\textwidth]{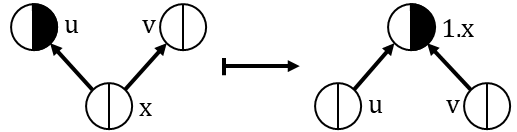}
    \caption{Basis case.}
    \label{subfig : informal basis case}
\end{subfigure}
\begin{subfigure}{\textwidth}
    \centering
    \includegraphics[width=\textwidth]{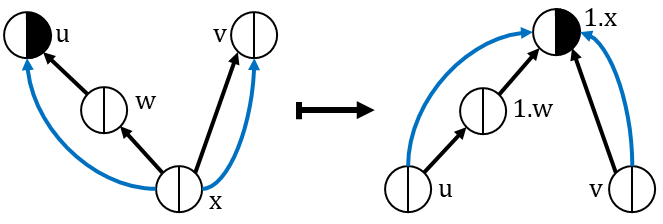}
    \caption{Easy case of dilation.}
    \label{subfig : informal bar}
\end{subfigure}
\begin{subfigure}{\textwidth}
    \centering
    \includegraphics[width=\textwidth]{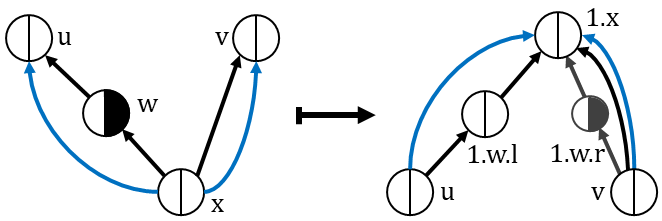}
    \caption{Hard case of dilation.}
    \label{subfig : creating gray vertices}
\end{subfigure}
\end{minipage}
\hfill
\begin{minipage}{0.6\textwidth}
\begin{subfigure}{\textwidth}
    \centering
    \includegraphics[width=0.8\textwidth]{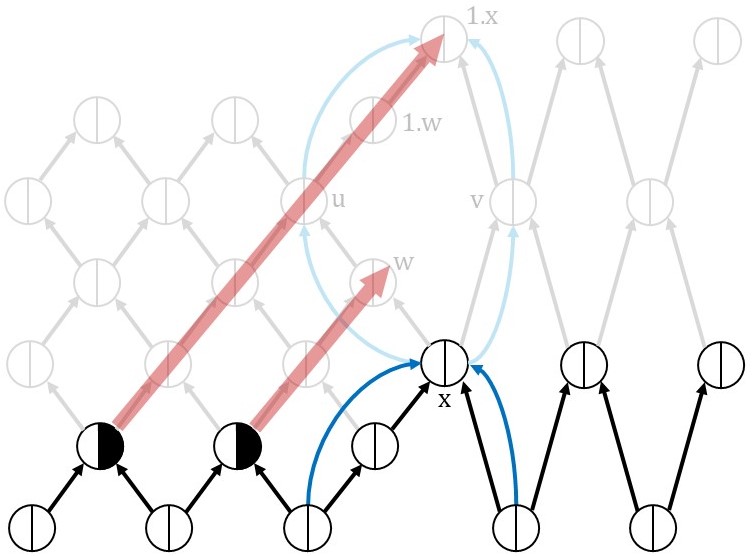}
    \caption{In black we describe an initial configuration with two right-moving particles. In grey we represent all possible rewritings obtained by applying the local operator described in~(\ref{subfig : informal basis case}) and~(\ref{subfig : informal bar}). In red we highlight the ``trajectories'' of the two particles, i.e. all possible vertices where they could end up after some rewriting steps.}
    \label{subfig : non reversible ST diagram}
\end{subfigure}
\end{minipage}
\caption{\emph{Time-dilation and reversibility.} $(a)$ The local operator acts by consuming the particle at $u$, thereby moving this particle to position $x$. It also flips the arrows pointing to $x$, and increments its timetag, in order to move the vertex from
past $x$ to future $1.x$. $(b)$ Here we have the same behaviour except for the fact that dilation edges force to update twice as much vertices on the left of $x$ than on its right. $(d)$ These two rules together generate the trajectory of the particle passing by $u$ and $w$. However it seems hard to extend the trajectory of that passing by $w$. $(c)$ One solution to this problem is to create a fresh gray vertex in between $v$ and $1.x$ to store the information that was at $w$.}
\label{fig : informal description of time dilation}
\end{figure}

\noindent\emph{Time-dilation in~\cite{ArrighiDetRew}.} Let us first recall the previous construction, and explain why it is not reversible. Internal states are pairs of bits $\sigma_G(u) = (\sigma^l_G(u),\sigma^r_G(u))\in\Sigma=\{0,1\}^2$, representing the presence of a left-moving particle or not, and of a right-moving particle or not. The set of \emph{ports} $\pi = C\times D = \{ b,d\}\times \{l,r\}$ is a cartesian product of a set of ``edge colors'' $C$ and a set of ``directions'' $D$. Each edge goes either from port $:\!(c,l)$ to port $:\!(c,r)$, or from port $:\!(c,r)$ to port $:\!(c,l)$, for $c\in C$. Each edge can thus be understood as having a colour $c$ and a spatial direction (``left'' or ``right''). Colours are used to distinguish two kinds of edges : \emph{standard edges} (using letter $b$, shown in black) and \emph{dilation edges} (using letter $d$, shown in blue). The local operator is informally summarized in Fig.~\ref{fig : informal description of time dilation}. In most cases it just makes particles move to the right or to the left (see Fig.~\ref{subfig : informal basis case}). But some vertices are equipped with two dilation edges, we call them `bar vertices'---i.e. a bar vertex is a vertex which has both its ports $:\!(d,l)$ and $:\!(d,r)$ occupied\footnote{This `bar vertices' are encoded slightly differently in~\cite{ArrighiDetRew} : internal states are used to mark them instead of dilation edges.}. When the local operator is applied to at a bar vertex (see Fig.~\ref{subfig : informal bar}), the dilation edges enforce that vertices on the left be updated twice as fast as vertices on the right. At the global level, this generates a space-time background which features a left/right time dilation, as illustrated by Fig.~\ref{subfig : non reversible ST diagram}. As such, this local rule is not reversible. This is because the space-time region on the right is twice less dense as that on the left, and thus some particles will be lost as they travel from left to right. Indeed consider the two particles passing by $u$ and $w$ in Fig.~\ref{subfig : non reversible ST diagram}. The one at $u$ does have a rectilinear trajectory to go along with in the diagram, but it is harder see where the one at $w$ is supposed to go. We could try to make it `bounce' when it arrives at the bar vertex, 
but this would still be non-reversible, as it could still end up overlapping with some left-moving particle, in some more complex scenario.

\noindent {\em Reversible time-dilation idea.} We solve this issue by making use of the name algebra to create a ``gray'' vertex, whose job is to carry the particle in the less dense part of the space-time region, see Fig.~\ref{subfig : creating gray vertices}. 
Considering arbitrary complex situations, with consecutive time-dilations, forces us to allow for arbitrary chains of gray vertices in between black vertices.

\noindent\emph{Defining the working set of graphs.} In order to formally encode these ``gray vertices'', we first enlarge the set of ports by adding a new edge color $g$ (hence we now have $C :=  \{ b,d, g\}$), and call `gray' a vertex which has both its gray ports $:\!(g,l)$ and $:\!(g,r)$  occupied. We then restrict our example to $\sshift$, the largest subset of $\mathcal{G}$ such that :
\begin{description}
    \item{{\em Non-gray vertices:}} each non-gray vertex has its two black ports occupied.
    \item{{\em Dilation edges:}} dilation edges are always connected to a bar vertex and a non-bar vertex. Moreover bar vertices are always connected, by their dilation edges, to a standard-edge-distance $1$ vertex on their right, and to a standard-edge-distance $2$ vertex on their left (this can be generalised but we wish to keep this first example relatively simple).
    \item{{\em Gray vertices:}} all gray vertices form finite chains, that start a the source of a black edge, and end at the target of that same black edge. 
    The chains must be of length $2^n$ for some integer $n$, and must carry at least one particle at an odd position.
    \item{{\em No-border:}} we consider graphs without borders, so that the closure properties of Def.~\ref{def : closure} be trivially respected.
\end{description}

\begin{figure}[t]
\centering
\begin{subfigure}{0.49\textwidth}
    \centering
    \includegraphics[width=\textwidth]{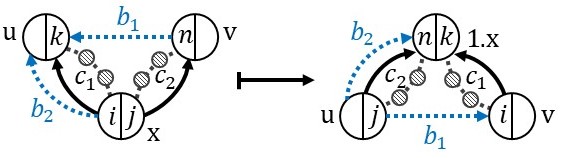}
    \caption{$\{D_x^{1i} \to E_x^{1i}\}_{x\in\X,j\in J}$}
\end{subfigure}
\begin{subfigure}{0.49\textwidth}
    \centering
    \includegraphics[width=\textwidth]{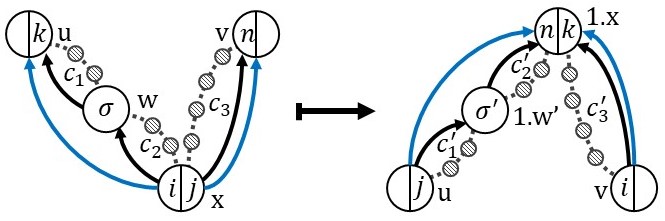}
    \caption{$\{D_x^{1i} \to E_x^{2i}\}_{x\in\X,j\in J}$}
\end{subfigure}
\begin{subfigure}{0.49\textwidth}
    \centering
    \includegraphics[width=\textwidth]{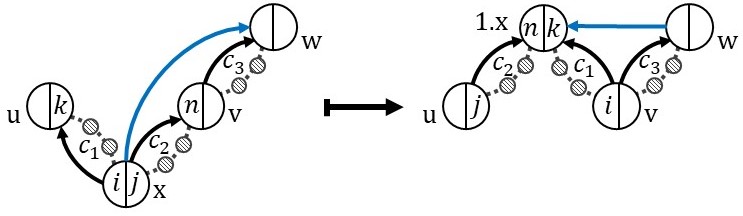}
    \caption{$\{D_x^{3i} \to E_x^{3i}\}_{x\in\X,j\in J}$}
\end{subfigure}
\begin{subfigure}{0.49\textwidth}
    \centering
    \includegraphics[width=\textwidth]{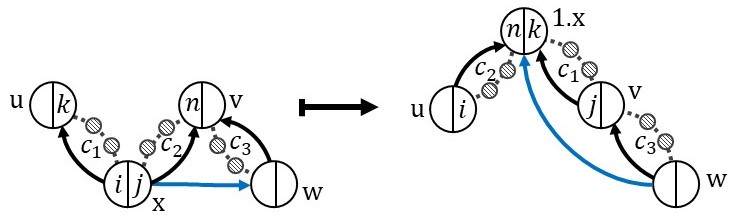}
    \caption{$\{D_x^{4i} \to E_x^{4i}\}_{x\in\X,j\in J}$}
\end{subfigure}
\caption{\emph{Causal rewriting system for reversible time dilation.} 
Whenever a cell (or half cell) is represented empty, the local rule does not change it. Gray chains are optional; if present they are of length $2^n$ and carry a  particle at an odd position. In $(a)$ there are two optional dilation edges (blue dotted lines). 
In $(b)$ the chains $c_1'$ and $c_2'$ and the internal state $\sigma'$ are obtained by ``cutting'' the chain $c_3$ in two pieces so as to propagate its particles in a rectilinear manner. The chain $c_3'$ is obtained by merging $c_1$, $c_2$ conversely.}
\label{fig : local operator}
\end{figure}

\noindent\emph{Defining the local operator.} We can then define the local operator as the causal rewrite system :
$$\{D_x^{kj}\to E_x^{kj}\}_{x\in\X,k\in \llbracket 1,4\rrbracket ,j\in J}$$
where each case $k$ is depicted in Fig.~\ref{fig : local operator}:
\begin{description}
    \item{{\em Case $k=1$:}} the basic case is just to propagate left and right moving particles, in a rectilinear manner. For this purpose, gray chains are just copied across, thereby smoothly extending the trajectories of particles they carry. 
    \item{{\em Case $k=2$:}} Around a bar vertex, gray chains need be created, merged or split in order to preserve the rectilinear trajectory of the particles, whilst ensuring reversibility. This case is further explained in Fig.~\ref{fig: ST diagramm complet}.
    \item{{\em Case $k=3,4$:}}
    After applying case $k=2$, we need to ``reload'' the left of the vertex bar, whilst preserving its left dilation edge. This will happen thought the sequence of applications of $k=3, 1, 4$. Notice that throughout the sequence, the vertex bar is non-past, so that $k=2$ cannot apply.
\end{description}

\begin{figure}[H]
    \centering
    \includegraphics[width=0.8\linewidth]{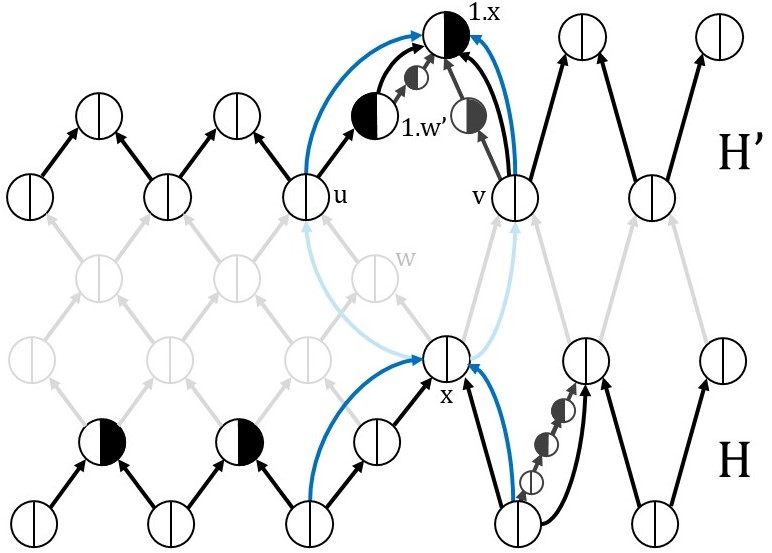}
    \caption{\emph{Reversible time dilation example.} In black we highlight two graphs : $H$ and one of its possible rewritings, $H'$. We start with one right moving particle on the left of the bar and a gray chain containing two particles on its right. After passing through the bar, the particle on the left ends up being stored in a gray vertex. In the other direction, the gray chain is split. Here one of its particles is stored in a black node, and an other in a new smaller gray chain, so that trajectories remain rectilinear. Note that we remove the empty gray.}
    \label{fig: ST diagramm complet}
\end{figure}

\noindent\emph{Reversibility.} It is not too hard to see that this defines a causal rewriting system, although will not provide these details here. Thus Prop.~\ref{lemma : charac mutex} tells us that we have defined a local operator $A_{(-)}$. But, is this local operator reversible ?\\ The easiest way to answer this question, is to check whether $\{{E_x^{kj}}\to {D_x^{kj}}\}_{x\in\X,k\in \llbracket 1,4\rrbracket,j\in J}$ is a backwards causal rewrite system and then apply Prop.~\ref{prop : mutex reversibility}. 
Indeed, it is not too hard to verify that this causal rewrite system is functional, renaming-invariant, context-preserving and $S$-preserving. 
However proving that it has a mutex domain---i.e. that ${\cal E}_x=\{{(E_x^{kj})}\}_{k\in \llbracket 1,4\rrbracket,j\in J}$ is a mutex set of backwards cone---is more tricky. Luckily this set of output graphs happens to be the exact transposition of the set of inputs disks $\D_x= \{D_x^{kj}\}_{k\in \llbracket 1,4\rrbracket,j\in J}$, i.e. $\transpose{{\cal E}_x} = \D_x$. Since we have already admitted that $\D$ is a mutex set of cones, and since $\transpose{\sshift}=\sshift$, we have that $x\mapsto{\cal E}_x$ defines a mutex set of backwards cone.


In fact this example is time-symmetric, and so a simple explicit description of the inverse local rule $B_{(-)}$ can be given. Indeed for all graph $G\in \sshift$ and $x\in \Past(G)$, we have that 
$$\transpose{(A^-_x\transpose{(A_x G)})} = G$$
where $A^-_x$ is the same dynamic than $A_x$ except that it decreases time tags. Thus the inverse local rule is $B_x=\transpose{~}\circ A^-_x\circ \transpose{~}$.

\noindent\emph{Commutativity.} Observe that two disks $G_{\restri{x}}$ and $G_{\restri{y}}$ can intersect at most on a single vertex. This makes it relatively straightforward to verify that the local rule $A_{(-)}$ is commutative (see Sec.\ref{sec:commutation}). Given that $A_{(-)}$ is also time-symmetric, we can then apply Lem.\ref{proof : time-symmetric commutation} to conclude that $B_{(-)}$ is also commutative as well.

\section{Conclusion}\label{sec:conclusion}

{\em Summary of results. } We use the framework of space-time deterministic graph rewriting in order to provide a rigorous notion of space-time reversibility. To do so we first combine the name algebra that is required for node creation/destruction, with the time tags required for space-time determinism, by letting both aspects commute.\\
The paper then just recalls the rest of framework, namely 1/ the definition of port graphs that is common in distributed computation but restricted to being a DAG in order to capture dependencies between events and 2/ the definitions of locality and renaming-invariance that a local operator $A_{(-)}$ must obey.\\ 
As a preparatory contribution we then introduce a notion of causal rewrite systems, which consists in searching for mutually exclusive patterns in the graph, and replacing them by others in a functional manner. We establish the equivalence between rewriting $G\to_x H$, and applying the local operator $H=A_x G$.\\
We are then set to provide three definitions of reversibility. The first demands that $A_{(-)}$ has an inverse that is itself a local operator $B_{(-)}$. Whilst a natural definition, this is asking for many properties at once, which is hard to check or construct. The second is axiomatic: it focuses on $A_{(-)}$ and demands just three high-level conditions be checked. The third is constructive: it provides a concrete way to enumerate these $A_{(-)}$, by ensuring that the corresponding rewrite system is causal both ways. All three definitions are proven equivalent, this is our core theoretical contribution.\\
We then show that the space-determinism of $A_{(-)}$, will be inherited by $B_{(-)}$ and even blending $A_{(-)}$ and $B_{(-)}$, so long as $B_{(-)}$ is also commutative. We identify three cases when $B_{(-)}$ inherits the commutativity of $A_{(-)}$, namely whenever $A_{(-)}$ is time-symmetric, two-two, or bounded.\\
Finally we turn to the construction of an example which is both reversible, space-time deterministic, and authentically asynchronous because it features time-dilation. The challenge here is not to loose information when a signal moves from a fast-ticking, fine-grained region, to a slow-ticking, coarse-grained region.

{\em Open questions. } In its current state our work leaves a number of technical questions open. For instance, the example we developed involves some non-finitary steps. We have good reasons to think that by endowing the graphs with a theory of equality inspired by discrete versions of Lorentz transformations~\cite{ArrighiLorentz} this could likely be avoided. But to make this fully rigorous would require developing a theory of reversible graph rewriting modulo an equality theory, in the style of~\cite{RewritingModulo}. \\
A more basic question is whether our definition of reversibility can be relaxed so that $A_x$ does not systematically send $x$ past to $x$ future.\\
Also, even if we have proven full consistency is inherited from $A_{(-)}$ to its commutative inverse $B_{(-)}$, we may wonder if the sufficient conditions identified in \cite{ArrighiDetRew} for $A_{(-)}$ to be fully consistent, such as `privacy', are also passed on to $B_{(-)}$.\\
In~\cite{ArrighiDetRew} we prove that space-time deterministic graph rewriting is capable of simulating any cellular automata. Is it the case that reversible graph rewriting is capable of simulating any reversible cellular automata? Any Reversible Causal Graph dynamics? In particular, in~\cite{ArrighiRCGD}, it was shown that Reversible Causal Graph Dynamics are implementable as a product of commuting local involutions---the similarity between these and the local operators of this paper is intriguing and deserves investigation.\\

{\em Comparison with other works. } 
A related strand of work is that surrounding `causal consistency' in the field of process algebra~\cite{KrivineDanosCausalConsistency}. Reversibility there is not Physics-inspired. Rather it is concerned with the Computer-oriented possibility of `undoing' and is obtained by having computational processes keep track of their history; which is not something that physical processes do. Still, causal consistency makes precise that actions which have not had further consequences (in the sense of other actions depending on it) can be undone, and that otherwise one can recursively first undo these consequences and then undo the action. This is clearly a property we have here---the connection that deserves further investigation.\\
Geometry is dynamical in our work. We are aware that the dominating vocabulary to describe them is now that of Category theory~\cite{LoweAlgebraic,Taentzer,HarmerFundamentals}.
We instead use the vocabulary of dynamical systems, but we are confident that abstracting away the essential features of our formalism could yield interesting categorical frameworks, e.g. à la~\cite{Maignan,ReversiblePawel,ReversibleHarmer}.

{\em Relativistic Physics analogy.} In General Relativity the primary object is 4D space-time, but it can still be cut into successive 3D slices which can be understood as successive snapshots aka `space-like cuts'. These snapshots are very much alike the spatial configurations of dynamical systems, but with the added subtlety that the slicing is allowed be done in an irregular manner. This is the so-called `time covariance' symmetry. It entails that it is perfectly legitimate to evolve just a small region of space, whilst keeping the rest of it unchanged, as modelled by asynchronism. Reversibility in this physical context can then be understood as the idea that any two closeby snapshots, must mutually determine each other. This idea is often just discussed however. Here we aimed to make it rigorous in a Discrete Mathematics, Computer Science setting.

{\em Perspectives.}  We hope that the notion of reversibility we developed may eventually find concrete applications for distributed computing, in terms of reducing power consumption; debugging/reproducibility; transactions rollbacks; modelling reversible chemical/biological reactions.\\
An interesting prospect is to understand whether space-times of reversible local operators correspond to expansive graph subshifts~\cite{ArrighiSubshifts}, in the same way that space-times of reversible cellular automata correspond to expansive tilings. One of the difficulties in establishing such a result is that the naming conventions we adopt for our vertices systematically prevent our space-times from being cyclic, even when the dynamics described is periodic---whereas the corresponding graph subshift will in fact be cyclic.\\
We, on the other hand, are likely to focus on the quantum regimes of these graph rewriting models.

\begin{credits}
\subsubsection{\ackname} This work has been partially funded by the European Union through the MSCA SE project QCOMICAL, by the French National Research Agency (ANR): projects TaQC ANR-22-CE47-0012 and within the framework of ``Plan France 2030'', under the research projects EPIQ ANR-22-PETQ-0007, OQULUS ANR-23-PETQ-0013, HQI-Acquisition ANR-22-PNCQ-0001 and HQI-R\&D
ANR-22-PNCQ-0002, and by the ID \#62312 grant from the John Templeton Foundation, as part of the \href{https://www.templeton.org/grant/the-quantum-information-structure-of-spacetime-qiss-second-phase}{‘The Quantum Information Structure of Spacetime’ Project (QISS) }. The opinions expressed in this project/publication are those of the author(s) and do not necessarily reflect the views of the John Templeton Foundation.

\end{credits}
%
%

\bibliographystyle{splncs04}
\bibliography{biblio}

\appendix

\section{On renaming-invariance and name-preservation}\label{sec:name-preservation}

\begin{lemma}[Renaming expressivity]\label{lemrenamingexpressivity}
    Let $G\in\calG$ a finite graph. There exists a renaming $R$ such that $V_{RG}\PA{\subseteq} \mathbb{N}$.
\end{lemma}
\begin{proof}\label{proof : expressivity of renaming}
    Suppose that $G$ has nodes $u, v,\ldots$ in which $m.p, m.q,\ldots$, occurs. We first show how to construct $R$ so that $m.p$ becomes $k$, $m.q$ becomes $l$,\ldots, with all the rest kept unchanged.\\
    Indeed, consider some fresh integers $k, l,\ldots$ that do not occur in anywhere in $V_G$. This is always possible because $G$ is finite and node names are finite trees. Build a new name $u$ out of fresh integers and the $\vee$ operator, so that $u.p=k$ and $u.q=l$. This is always possible because due `non-overlapping positions' condition of Def.~\ref{def : graphs}, $p$ cannot be a prefix of $q$ and reciprocally. Consider $R$ which acts as the identity on all of the integers occurring in $V_G$, except for $m$ which is sent to $u$. Complete $R$ to be a bijection over $\X$. Extend $R$ to act over $\calV$. This is the appropriate renaming since $R(m.p)=R(m).p=u.p=k$, $R(m.q)=R(m).q=u.q=l$\ldots\\
    If $RG$ has other nodes $u', v',\ldots$ in which $m'.p', m'.q',\ldots$, occur, we again build $R'$ and consider $\cdots R'RG$ until we reach a graph $H$ whose names are made of solely out of integers and the $\vee$ operator.\\
    Next suppose that $H$ has a node $u$ in which $m\vee n$, occurs. We show how to construct $S$ so that $m\vee n$ becomes $m$, with all the rest kept unchanged.\\
    Indeed consider $S$ which acts as the identity on all of the integers occurring in $V_G$, except for $m$ which is sent to $m.l$ and $n$ which is sent to $m.r$. Complete $S$ to be a bijection over $\X$. Extend $S$ to act over $\calV$. This is the appropriate renaming since $R(m\vee n)=R(m)\vee R(n)=m.l\vee m.r=m$. It does not act anywhere else, since by the non-overlapping positions' condition of Def.~\ref{def : graphs}, the integers $m$ and $n$ did not occur anywhere else in $H$.\\
    If $SH$ has other nodes $u$ in which $m'\vee n'$, occur, we again build $S'$ and consider $\cdots S'SG$ until we reach a graph $G'$ whose names are made of solely out of integers.\\
    The renaming of the lemma is $(\cdots S'S\cdots R'R)$, it does map $G$ into a $G'$ whose names belong to $\mathbb{N}$.
\end{proof}

\begin{proposition}[Name-preservation]\label{prop:renaminginvarianceimpliesnp}~\\
Let ${F_x}:{\sshift}\to \calV^*$ be a renaming-invariant function, then $F_x(G) \namesubseteq V_G \cup \{x\}$.\\
Let ${A_x}:\sshift\to\sshift$ 
be a renaming-invariant bijective function, then $V_{A_x G}\cup \{x\} \nameeq V_{G}\cup \{x\}$.\\
Let $A_x:{\sshift}\to {\sshift}$ be a local operator, then $V_{A_x G} \namesubseteq V_G$.\\
Let $A_x:{\sshift}\to {\sshift}$ be a reversible local operator, then $V_{A_x G}\nameeq V_G$.
\end{proposition}
\begin{proof}
$[\textrm{Renaming-invariance implies }F_x(G) \namesubseteq V_G\cup\{x\}]$ \\
By contradiction. Say $v\in \namealg{F}_x(G)\text{ \ and \ } v\notin \namealg{V_G \cup \{x\}}$. 
Pick $R$ such that $RG=G$, $R(x)=x$, $R(v) \notin \namealg{F}_x(G)$, i.e. map $v$ into a fresh name $v'$ whilst preserving $x$ and $G$. We have:
\begin{align*}
RF_x(G) =& F_{\sp{R( x)}}(RG)\quad\text{by renaming-invariance.}\\
=& F_{x}(G)\quad\text{by choice of $R$.}
\end{align*}
There are infinitely many such $R$, and since $v\in \namealg{F}_x(G)$, there are infinitely many such $RF_x(G)$. It follows that $F_{x}(G)$ is not determined, hence the contradiction. 
The result follows, from which we also have that ${V_{A_x G} \cup \{x\}} \namesubseteq {V_G \cup \{x\}}$.

$[\textrm{Renaming-invariant bijective implies }V_{A_x G}\cup \{x\} \nameeq V_{G}\cup \{x\}]$ \\
By Lem.~\ref{lem:inverserenamings}, $A_{x}^{-1}$ is also renaming-invariant.\\
So, the same reasoning applies. 
We therefore have that ${V_{A_x G} \cup \{x\}} \namesubseteq {V_G \cup \{x\}}$.

$[$\textrm{Locality implies }$V_{A_x G} \namesubseteq V_G$ $]$~\\
If $G\notin \graphvalid{\restri{}}{x}$, $A_xG=G$ makes it trivial. So we consider $G\in  \graphvalid{\restri{}}{x}$ and have $A_xG=(A_x G_\restri{x}) \sqcup G_\crestri{x}$. For this $\sqcup$ to be defined we need $B_{A_x G_\restri{x}}=B_{G_\restri{x}}$. Thus, for any new name $v\in \namealg{V}_{A_x G}\setminus\namealg{V}_G$ to exist, there must be some $w\in \namealg{I}_{A_x G_\restri{x}}\setminus\namealg{V}_{G}$. This means there exists $p\in\{l,r\}^*$ such that $u=w.p$ and $u\notin \namealg{V}_{G}$.\\
Construct a graph $G'\in \sshift$ such that $ I_G\cup \{u\} \subseteq I_{G'}$ and $G\sqsubseteq G'$. Such a graph can exist since $u\notin\namealg{V}_{G}$, but we must make sure that $G\in \sshift$. This can be done by first constructing $RG$ a copy of $G$ which contains only fresh names, including $u$. This $RG\in \sshift$ since $\sshift$ is closed under renaming. Second, since $\sshift$ is closed under disjoint union, we can take $G'= G\sqcup RG$.
Note that we still have $G_{\restri{x}}\sqsubseteq G'$, and so by strong extensivity, $G'_\restri{x}=G_\restri{x}$, from which it follows that $A_xG'=(A_x G_\restri{x}) \sqcup G'_\crestri{x}$. But this $A_xG'$ is undefined as $w\in I_{A_x G_\restri{x}}$ and $w.p=u\in G'_\crestri{x}$ violates the  non-overlapping positions condition.

$[$\textrm{Reversibility implies }$V_{A_x G} \namesubseteq V_G$ $]$~\\
By reversibility there is a local operator $B_x$ such that $G=B_{x}(A_x G)$. Using twice the above, we have ${V_G}=V_{B_xA_xG} \namesubseteq V_{A_x G} \namesubseteq V_G$.
\end{proof}

\section{On causal rewrite systems}\label{app:mutex}

\begin{lemma}[Mutex cones define neighbourhoods]\label{lemma : neighbourhood from 0 to x}
    Let \PA{$\mathcal{D}:x\mapsto \D_x$ define} a mutex set of cones. Let $\restriD{\D}{}$ be the function that maps any $x\in \mathcal{X}$ to a function 
    $\restriD{\mathcal{D}}{x}:{\graphvalidmutex{\D}{x}} \to P(\mathcal{X})$, $G\mapsto\restrD{\mathcal{D}}{x}{G}$, with $\restrD{\D}{x}{G}$ as given by Def.~\ref{def : mutex cones}. 
    This $\restriD{\mathcal{D}}{}$ is a well-defined neighbourhood scheme.
\end{lemma}
\begin{proof}
{\em Renaming-invariance.} We want $R\restrD{\mathcal{D}}{x}{G} = \restrD{\mathcal{D}}{\sp{R(x)}}{R G}$. 
Say that there exists $D\in \D_x$ such that $D\sqsubseteq G$. This is equivalent, by the fact that $\D$ is renaming-invariant and $\sshift$ is closed under renaming, to the existence of $RD\in \D_{\sp{R(x)}}$ such that $RD\sqsubseteq RG$. Say $D$ and $RD$ exist, by unambiguity they are unique and we have by definition that 
$D=G_\restriD{\mathcal{D}}{x}$ and $RD=(RG)_\restriD{\mathcal{D}}{\sp{R(x)}}$.
Then $R G_\restriD{\mathcal{D}}{x}=RD=(RG)_\restriD{\mathcal{D}}{\sp{R(x)}}.$ Thus,
$R\restriD{\mathcal{D}}{x}(G)=\restriD{\mathcal{D}}{\sp{R(x)}}(RG)$ as requested.
Then we check the properties of Def.\ref{def:neighbourhood}.\\
\emph{Completeness.} 
\PA{We want to show that for all $G\in\upast{u}\maxi$ \RX{with $x=\sp{u}$,}{there exists $x\in\X$ with $x\namesubseteq\sp{u}$ such that} we have $G\in\graphvalidmutex{\D}{x}$.} 
This comes from the completeness condition of a mutex set of cones. Indeed consider $G\in \upast{u}\maxi$. There exists a \maxgraph cone $C\in \xCone\maxi$ such that $C\sqsubseteq G$. Completeness then gives us the existence of some $D\in \D_x$ such that $D\sqsubseteq C\sqsubseteq G$, with \RX{$x=\sp{u}$}{$x\namesubseteq\sp{u}$} \PA{as $\D$ is well-indexed}. This proves $G\in \graphvalidmutex{\D}{x}$.\\
\RX{}{\noindent {\em Unambiguity.} We want that $\forall G\in \upast{u},\ |\{x\in\X \mid x\namesubseteq\sp{u}\textrm{ and }G\in\graphvalidmutex{\D}{x}\}|\leq 1$.
Consider $G\in \upast{u}$ and let $C$ be the cone obtained by starting from $u$ and exploring $G$ along the directed edges. Recall that $G\in\graphvalidmutex{\D}{x}$ is equivalent to the existence of $x\in\X,D\in\D_x$ with $D\sqsubseteq G$. So this implies the existence of $x\in\X,D\in\D_x$ with $D\sqsubseteq C$. But since $\D$ is unambiguous we know that $|\{D\in \xcone{} \mid \exists x\in\X,D\in \D_x\textrm{ and } D\sqsubseteq C\}|\leq 1$.}

\noindent\emph{Cone.} Let $x\in \mathcal{X}$ and $G\in \graphvalidmutex{\mathcal{D}}{x}$. \PA{So there exists a unique cone $D\in {\cal D}_x$ such that $D\sqsubseteq G$. We have $G_{\restriD{\D}{x}}= D$ which is a cone of $x$.}\\
\emph{Strong extensivity.} Let $x\in \mathcal{X}$ and $G\in \graphvalidmutex{\mathcal{D}}{x}$. 
Let $H\in \sshift$ be a graph such that $G_\restriD{\mathcal{D}}{x}\sqsubseteq H$ we have to prove $\restrD{\mathcal{D}}{x}{H} = \restrD{\mathcal{D}}{x}{G}$. This follows from the fact that $G_\restriD{\mathcal{D}}{x}$ occurs in $H$ and by unambiguity, is the only such graph.
\end{proof}

\lemmamutexneighbourhood*
\begin{proof}
    \emph{Neighbourhood schemes yield mutex sets.} We consider the set of cones $\diskA{}$. \PA{By the cone condition each $G_\restri{x}$ is an element of $\xcone{x}$, so $\diskA{x}\subseteq\xcone{x}$ and $\diskA{}$ is well-indexed.} It is renaming-invariant because $\restriA{}$ and is renaming invariant and $\sshift$ is closed under renaming. To establish the completeness of $\diskA{}$, we consider $C\in \xcone{}\maxi$ and we will prove \PA{$\exists x\in\X,\exists D\in \diskA{x}, D\sqsubseteq C$}. 
    We consider the vertex $u\in \Past(C)$ and we note $x=\sp{u}$. Since $C\in \upast{u}\maxi$, by completeness of $\restri{}$ we have $C\in \graphvalid{\restri{}}{x}$. Then $C_{\restri{x}}$ is defined, and we have $C_{\restri{x}}\in \diskA{x}$ and $C_{\restri{x}}\sqsubseteq C$. 
    To prove unambiguity consider an $H\in \sshift$ with two disk occurrences, i.e. such that there exists $G_{\mathcal{M}_x},G'_{\mathcal{M}_x}\in{ \sshift}_{\mathcal{M}_x}$ with $G_{\mathcal{M}_x} \sqsubseteq H$ and $G'_{\mathcal{M}_x} \sqsubseteq H$.
    Then by strong extensivity we have $G_{\mathcal{M}_x} = H_{\mathcal{M}_x}$ and $G'_{\mathcal{M}_x} = H_{\mathcal{M}_x}$, which implies $G_{\mathcal{M}_x} = G'_{\mathcal{M}_x}$.

    \emph{Mutex sets yield neighbourhood schemes.} To any mutex set of cones \PA{$\D:x\mapsto\D_x$}, we are associating a neighbourhood scheme $\restriD{\D}{}$ as in Lem.~\ref{lemma : neighbourhood from 0 to x}. It remains to show that is is a bijection.

    \emph{Right inverse.} Let $\restriB{}$ be a neighbourhood scheme, we want $\restriB{} = \restriD{\sshift_{\restriB{}}}{}$. Let $x\in \mathcal{X}$ be a position. First we prove that $\graphvalid{\restriB{}}{x}$ the domain of $\restriB{x}$ is equal to $\graphvalidmutex{\diskB{}}{x}$ the domain of $\restriD{\diskB{}}{x}$. Let us unravel the definitions. On one side $\graphvalidmutex{\diskB{}}{x}$ is the set of all graph $G$ such that there exists $D\in \diskB{x}$ such that $D\sqsubseteq G$. For all $G\in \graphvalid{\restriB{}}{x}$ we have $G_{\restri{x}}\sqsubseteq G$ and so $\graphvalid{\restriB{}}{x} \subseteq \graphvalidmutex{\diskB{}}{x}$. Next, by strong extensivity $\restriB{}$ is defined on all graphs containing a disk, so that we also have $ \graphvalidmutex{\diskB{}}{x}\subseteq \graphvalid{\restriB{}}{x}$.\\
    Then we consider $G\in \graphvalid{\restriB{}}{x}=\graphvalidmutex{\diskB{}}{x}$. By definition there exists $D \in \diskB{x}$ such that $D \sqsubseteq G$. 
    We then we have, by strong extensivity $G_{\restriD{\diskB{}}{x}}=D=G_{\restriB{x}}$
    from which $\restriB{} = \restriD{\sshift_{\restriB{}}}{}$ follows.

    \emph{Left inverse.}
    Now we consider \PA{$\D:x\mapsto\D_x$} a mutex set of cones, and we prove that ${\sshift}_{\restriD{\D}{}} = \D$. \PA{First we check that $\D_x\subseteq {\sshift}_{\restriD{\D}{x}}$. 
    By definition of $\restriD{\D}{x}$, we have $D_{\restriD{\D}{x}} = D$ which implies $D\in {\sshift}_{\restriD{\D}{x}}$. The reverse inclusion ${\sshift}_{\restriD{\D}{x}}\subseteq \D_x$ is by definition.} 
\end{proof}

\begin{lemma}[Full-exploration implies border-completion]\label{lemma : border completion}
    Any subset of graph $\sshift$ closed under \maximality and \bckwmaximality is such that $G\in \sshift$ implies that $\exists G'\in \sshift$ where  $V_G\subseteq I_{G'}$ and $G\sqsubseteq G'$.
\end{lemma}
\begin{proof}
    We get by \maximality a graph $H$ which does not contain any forward dangling edges. This means that $B_G\cap B_H$ only contains border vertices connected by backward dangling edges. Then we apply the \bckwmaximality condition on $H$. We get a graph $G'$ such that $G\sqsubseteq H\sqsubseteq G'$ with no backward dangling edges. This means that $B_{G'}\cap B_G=\emptyset$, which implies by inclusion $V_G\subseteq I_{G'}$.
\end{proof}

\propcharacmutex*
\begin{proof}

\emph{Local operator yield causal rewrite system.} Let $A_{(-)}$ be a $\restriA{}$-local operator. We consider the following rewrite system :
$$\{ D \to A_x D\}_{x\in\X,D\in \diskA{x}}$$
It is by definition functional. It has a mutex domain $\diskA{}$ by Lem.~\ref{lemma : equivalence set mutex and neighbourhood schemes}. Renaming invariance comes directly from the renaming invariance of $A_{(-)}$. Finally we show it is context-preserving. First $F_{A_xD}=F_{D}$ can be derived from locality and the border completion of $\sshift$ (Lem.\ref{lemma : border completion}). Second this entails $B_{A_xD}=B_{D}$, but we also have $V_{A_xD}\namesubseteq V_D$ by Prop.~\ref{prop:renaminginvarianceimpliesnp} (local operator case), and so $I_{A_xD} \namesubseteq {I}_{D}$ . 

\emph{Causal rewrite system yield local operator.} Let $\{D_x^j\to G_x^j\}_{x\in\X,j\in J}$ be a deterministic rewrite system. We consider the neighbourhood scheme $\restriA{}^{\mathcal{D}}$ defined in Lem.~\ref{lemma : equivalence set mutex and neighbourhood schemes} with $\D:x\mapsto\{D_x^j\mid j\in J\}$. Note that $\restriD{\D}{x}$ admits as domain $\graphvalidmutex{\D}{x}$, i.e. the set of all graph $G$ which can be written $G= D_x^j \sqcup H$. Let us show that the function $A_{(-)} : \mathcal{X}\times {\sshift}\to {\sshift}$ such that $\forall G\in {\sshift}$, $\forall x\in \mathcal{X}$ :
$$A_x G := 
\begin{cases}
G_x^j \sqcup G_{\overline{\restriD{\D}{x}}}&\textrm{if $G\in \graphvalidmutex{\D}{x}$ with $G = D_x^j \sqcup G_{\overline{\restriD{\D}{x}}}$}\\
G&\textrm{otherwise}
\end{cases}
$$
is a well-defined $\restriD{D}{}$-local operator. Indeed $G_x^j \sqcup G_{\overline{\restriD{\D}{x}}}$ is a graph because of the context preserving condition, and this graph belongs to $\sshift$ because the causal rewrite system is $\sshift$-preserving. Moreover it is renaming invariant because $\sshift$ is closed under renaming, $\restriD{D}{}$ is renaming invariant by Lem.~\ref{lemma : neighbourhood from 0 to x} and $\{D_x^j\to G_x^j\}_{x\in\X,j\in J}$ is renaming invariant. Thus this $A_{(-)}$ is indeed a $\restriD{D}{}$-local operator.

\emph{Right inverse.} We prove that for any pair $(\restriA{}$,$A_{(-)})$
, mapping it to a rewrite system $\{ C \to A_x C\}_{x\in\X,C \in \sshift_{\restriD{}{x}}}$, and then mapping the obtained rewrite system back to a pair $(\restriD{\diskA{}}{}, A_{(-)}')$ as in the previous paragraph, acts as the identity. 
First we note that by bijectivity of the map between mutex sets and neighbourhood functions proven in Lem.~\ref{lemma : equivalence set mutex and neighbourhood schemes}, we have $\restriD{\diskA{}}{}=\restriA{}$. Then we just have to check that for all $x\in\mathcal{X}$ and $G= D_x^j \sqcup G_{\crestri{x}}$ we have $A_{x}' G = A_{x} G$. It is the case : 
$$A_{x}' G = G_x^j \sqcup G_{\crestri{x}} = A_x D_x^j \sqcup G_{\crestri{x}} = A_x G_{\restriA{x}} \sqcup G_{\crestri{x}}=A_x G.$$

\emph{Left inverse.} Now we consider a causal rewrite system $\{D_x^j\to G_x^j\}_{x\in\X,j\in J}$, and we prove that mapping it to a neighbourhood scheme and a local rule, and then mapping the obtained local rule back to a rewrite system, acts as the identity. Let $\D:x\mapsto \{D_x^j\}_{j\in J}$. We build the associated $\restriD{\D}{}$-local rule $A_{(-)}$. Note how this rule is such that $A_x D_x^j = G_x^j$. Consider its associated causal rewrite system :
    \begin{align*}
        \{ D \to A_x D\}_{x\in \X , D\in {\sshift}_{\restriD{\D}{x}}} &= \{ D \to A_x D\}_{x\in \X ,D\in \mathcal{D}_x}\tag{Lem.~\ref{lemma : equivalence set mutex and neighbourhood schemes}}\\
        &= \{ D_x^j \to G_x^j\}_{x\in\X,j\in J}
    \end{align*}
which concludes the proof.
\end{proof}

\section{On reversibility}\label{app:reversibility}

\lemuniqueness*\label{proof : uniqueness of the inverse}
\begin{proof}
    We must prove that for all $G\in \sshift$ and $x\in \mathcal{X}$ we have $B_x G = B'_x G$.
    If both $B_{x}$ and $B'_{x}$ act trivially on $G$ the result is immediate. Let us suppose that $B_{x}$ acts non trivially on $G$. This entails \PA{$G\in \graphvalid{\restriB{}}{x}$}. The right inverse property then implies \PA{$B_x G \in \graphvalid{\restriA{}}{x}$}. Thus we have 
        $B_x G = B'_x A_x B_x G= B'_x G$.
\end{proof}

\begin{lemma}[Inverse renamings]\label{lem:inverserenamings}
If $R$ is a renaming, so is $R^{-1}$.\\
Let ${A_{(-)}}$ 
be a renaming-invariant bijective function. If $A$ is renaming-invariant, so is $ A_{(-)}^{-1}$.
\end{lemma}
\begin{proof} \noindent (Adapted from~\cite{ArrighiQNT})\\
\noindent {\em Inverse renaming.} Let $R$ be a renaming. We need to check that $R^{-1}$ is a homomorphism of the name algebra.\\
For any $u',v'$ take $u,v$ such that $u'=R( u)$ and $v'=R( v)$.\\
$R^{-1}( u'.p) =R{^{\ }}^{-1}( R( u) .p) =R{^{\ }}^{-1}( R( u.p)) =u.p=R^{-1}( u') .p$\\
$R^{-1}( u'\lor v') =R{^{\ }}^{-1}( R( u) \lor R( v)) =R{^{\ }}^{-1}( R( u\lor v)) =u\lor v=R^{-1}( u') \lor R^{-1}( v')$.\\
$R^{-1}( t.u') =R{^{\ }}^{-1}( t.R( u)) =R{^{\ }}^{-1}( R(t.u)) =t.u=t.R^{-1}(u')$.

\noindent {\em Inverse renaming-invariance.}\\
$RA_{x}^{-1} =\left( A_{x} R^{-1}\right)^{-1} =\left( R^{-1} A_{\sp{R( x)}}\right)^{-1} =A_{\sp{R(x)}}^{-1} R$.
\end{proof}

\propcharacterizationreversibility*\label{proof:prop:characterizationreversibility}
\begin{proof}
$[\Rightarrow]$ Suppose axiomatic-reversibility. 
Let $\graphvalidexplicit{\restriB{}}{x}:=\{H=A_xG\mid G\in\graphvalidexplicit{\restriA{}}{x}\}$. 
Notice that for any such $H$, the corresponding $G$ is unique, otherwise this would contradict injectivity. Thus $A_x^{-1}$ is well-defined over $\graphvalidexplicit{\restriB{}}{x}$, and we can let $\restriB{x}:=\restriA{x}\circ A_x^{-1}$ over this domain. Notice that this function is renaming invariant by the renaming-invariance of $\restriB{x}$ and that of $A_x^{-1}$ which was proven in Lem.~\ref{lem:inverserenamings}.\\ 
Remark that for such an $H=A_xG$, by $\restriA{}$-locality of $A_x$ we have $H=(A_x G_{\restriA{x}})\sqcup G_{\crestriA{x}}$ and so $H_{\restriB{x}}=H_{\restrA{x}{G}}=A_x G_{\restriA{x}}$ and $H_{\crestriB{x}}=H_{\crestrA{x}{G}}=G_{\crestriA{x}}$.\\ 
Now we check that the renaming invariant function $\restriB{}$ is a neighbourhood scheme.\\
\emph{Cone.} $H_\restriB{x}=A_x G_\restriA{x}$ which by back-reachability is a backwards cone of $x$.\\
\emph{Completeness.} \PA{We want to prove that for all $H\in {}\maxi\!\ufut{u}$,} \RX{$x=\sp{u}$,}{there exists $x\in \X$ with $x\namesubseteq\sp{u}$ such that} we have $H\in\graphvalidexplicit{\restriB{}}{x}$. 
\PA{This comes from surjectivity.}
\emph{Strong extensivity.} Suppose $E\in \diskB{x}$ and $H\in \sshift$ such that $E\sqcup K=H$. We want to show that $H_{\restriB{x}} = E$. We will see that this comes from the strong extensivity of $\restriA{}$.\\
Since $E\in \diskB{x}$ there exists $H'\in \sshift$ such that $E\sqcup K'=H'$ and $E=(E\sqcup K')_\restriB{x}$. Let $G':=B_x (E\sqcup K')$. By $\restri{}$-locality of $A_x$ we have 
$$(E\sqcup K')=A_x G'= (A_x G'_\restri{x}) \sqcup G'_\crestri{x}.$$
But, since 
$$\restri{x}(G')\nameeq\restriB{x}(A_x G')=\restriB{x}(E\sqcup K')=I_E,$$
it must be the case that $E=A_x G'_\restri{x}=A_x D$ with $D:=G'_\restri{x}\in \diskA{x}$.\\
The second back-reachability hypothesis then tells us that $G= D\sqcup K\in \sshift$. By the strong extensivity of $\restri{}$, we have that 
$$A_x G = A_x (D\sqcup K) = (A_x D)\sqcup K= E\sqcup K = H.$$
Thus,
$$\restriB{x}(H)=\restriB{x}(A_x G)\nameeq\restri{x}(G)=I_D\nameeq I_{A_xD}=I_E$$
and so $H_{\restriB{x}} = E$.

\noindent {\em Locality.} We then let $B_x$ to coincide with $A_x^{-1}$ on $\graphvalidexplicit{\restriB{}}{x}$, and be the identity otherwise. 
This $B_x$ is $\restriB{x}$-local. Indeed consider $H\in\graphvalidexplicit{\restriB{}}{x}$ we have $H=A_x G$. Recall that by $\restriA{}$-locality of $A_x$ we have $H=(A_x G_{\restriA{x}})\sqcup G_{\crestriA{x}}$  and so $H_{\restriB{x}}=H_{\restrA{x}{G}}=A_x G_{\restriA{x}}$ and $H_{\crestriB{x}}=H_{\crestrA{x}{G}}=G_{\crestriA{x}}$. 
Let $f_x : A_x G_{\restriA{x}}\mapsto G_{\restriA{x}}$. Thus :
\begin{align*}
    B_x H &= B_x A_x G
    = G = G_{\restriA{x}}\sqcup G_{\crestriA{x}}\\
    &= f_x A_x G_{\restriA{x}}\sqcup G_{\crestriA{x}}\\
    &=f_x H_{\restriB{x}}\sqcup H_{\crestriB{x}}\textrm{ by the remark.}
\end{align*}

\noindent {\em Left local invertibility} follows from these definition: $\forall G\in\graphvalidexplicit{\restriA{}}{x}$, $B_xA_xG=A_x^{-1}A_xG=G$ and $\restrB{x}{A_xG}\nameeq\restrA{x}{A_x^{-1}A_xG}=\restrA{x}{G}$.\\
\noindent {\em Right local invertibility} also follows from these definition: $\forall H\in\graphvalidexplicit{\restriB{}}{x}$, $A_xB_xH=A_xA_x^{-1}H=H$ and $\restrB{x}{H}\nameeq\restrA{x}{A_x^{-1} H}=\restrA{x}{B_x H}$.

The renaming-invariance of $B_x$ follows from Lem.~\ref{lem:inverserenamings}.

\noindent$[\Leftarrow]$ Suppose reversibility.\\ 
\noindent {\em Injectivity.} Consider distinct $G,G'\in \graphvalidexplicit{\restriA{}}{x}$. Left local invertibility gives us $B_x A_x G = G\neq G' = B_x A_x G'$. Thus $A_x G \neq A_x G'$ which is injectivity.\\
\noindent {\em Surjectivity.} \PA{Consider $H\in {}\maxi\!\ufut{u}$. 
By completeness of $\restriB{}$ we have \RX{, with $x=\sp{u}$,}{that there exists $x\in \X$ with $x\namesubseteq u$ such that} $H\in\graphvalidexplicit{\restriB{}}{x}$.}\\
Right local invertibility gives us $A_x B_x H=H$ with $\restrA{x}{B_x H}\nameeq\restrB{x}{H}$. Let $G:=B_x H$, since $\restrA{x}{G}\nameeq\restrB{x}{H}$ we have $G\in\graphvalidexplicit{\restriA{}}{x}$ and $H=A_x G$, which is surjectivity.\\
\noindent {\em Back-reachability.} Finally, consider $D\in \diskA{x}$, by definition there exists $G'=D\sqcup K'\in\sshift$, and so $H'=A_x G'=A_x D\sqcup K'\in\sshift$. Left local invertibility and strong extensivity give $\restrB{x}{A_x G'}\nameeq\restrA{x}{G'}=I_D\nameeq I_{A_x D}$. By the fact that $\restrB{x}{A_x G'}$ is defined and backwards reachable, we deduce that $A_x D$ is a backwards cone of $x$. We also deduce that $B_x A_x D \sqcup K=B_x H'=G'=D \sqcup K$ and so as on-site functions, $B_x A_x D=D$. Now say that $H=A_x D\sqcup K\in \sshift$. By the strong extensivity of $\restriB{}$, $H_\restriB{x}=A_x D$, and so $B_x H=B_x A_x D \sqcup K=D\sqcup K\in \sshift$.
\end{proof}




\proprevastwowaycasaulrewsys*
\begin{proof}
 We call $\mathcal{E}_x = \{E_x^j\}_{j\in J}$ and $\mathcal{D}_x = \{D_x^j\}_{j\in J}$. We let $\restriA{} := \restriD{\D}{}$ 
 by Lem.~\ref{lemma : neighbourhood from 0 to x}).\\
    
\noindent \emph{$\{{E_x^j}\to {D_x^j}\}_{x\in\X,j\in J}$ causal implies $A_{(-)}$ reversible.} Since $\{{E_x^j}\to {D_x^j}\}_{x\in\X,j\in J}$ is a backwards causal rewrite system we can construct using Prop.~\ref{lemma : charac mutex} a unique backwards local operator ${B}_{(-)}$. 
We denote $\restriB{}=\restriD{\calE}{}$ its neighbourhood scheme coming from Lem.\ref{lemma : equivalence set mutex and neighbourhood schemes}. 
Let us prove that $B_{(-)}$ is the inverse of $A_{(-)}$.
    
    \emph{Left local invertibility.} 
    Let $x\in \mathcal{X}$. Let $G\in \graphvalid{\restriA{}}{x}$. Since $\restriA{}$ is defined on this graph it can be written $G = D_x^j \sqcup G_{\crestriA{x}}$. First we prove $\restrB{x}{A_x G} \PA{\nameeq} \restrA{x}{G}$. On the one hand $\restrB{x}{A_x G} = \restrB{x}{E_x^j \sqcup G_{\crestriA{x}}} = I_{E_x^j}$. On the other hand $\restrA{x}{G} =  I_{D_x^j}$. But by the context-preservation hypothesis, $I_{D_x^j} \nameeq I_{E_x^j}$, thus we have indeed $\restrB{x}{A_x G} \nameeq \restrA{x}{G}$. Then we derive :
    \begin{align*}
        B_x A_x G =& B_x A_x(D_x^j \sqcup G_{\overline{\restrA{x}{G}}})\\
        =& B_x ( E_x^j \sqcup G_{\overline{\restrA{x}{G}}})\\
        =& D_x^j \sqcup G_{\overline{\restrA{x}{G}}}\\
        =& G.
    \end{align*}
    
    \emph{Right local invertibility} is similar. 
   Let $x\in \mathcal{X}$. Let $G\in \graphvalid{\restriB{}}{x}$. Since $\restriB{}$ is defined on this graph it can be written $G = E_x^j \sqcup G_{\crestriB{x}}$. We have $\restrA{x}{B_x G} = \restrA{x}{D_x^j \sqcup G_{\crestriB{x}}} = I_{D_x^j} \nameeq I_{E_x^j} = \restrB{x}{G}$. Then we derive :
    \begin{align*}
        A_x B_x G =& A_x B_x(  E_x^j \sqcup G_{\overline{\restrB{x}{G}}})\\
        =& A_x ( D_x^j \sqcup G_{\overline{\restrB{x}{G}}})\\
        =&  E_x^j \sqcup G_{\overline{\restrB{x}{G}}}\\
        =& G.
    \end{align*}

    \emph{ $A_{(-)}$ reversible implies $\{{E_x^j}\to {D_x^j}\}_{x\in\X,j\in J}$ causal.} 
    We suppose the existence of a $\restriB{}$-local inverse $B_{(-)}$. 
    By Prop.~\ref{lemma : charac mutex}, ${B}_{(-)}$ is alternatively expressed by some backwards causal rewrite system $\{{{E_x^j}'}\to {{D_x^j}'}\}_{x\in\X,j\in J}$. In the following we prove that this rewrite system is equal to $\{{E_x^j}\to {D_x^j}\}_{x\in\X,j\in J}$.
    
    First we prove that $\mathcal{E}' = \{{{E_x^j}'}\}_{x\in\X,j\in J}$ and $\mathcal{E} = \{{E_x^j}\}_{x\in\X,j\in J}$ are equal. We start by proving $\mathcal{E}\subseteq\mathcal{E}'$. We pick $E_x^j\in \calE_x:=\calE\cap\xcone{x}$, and we prove $E_x^j\in \diskB{x}$ which implies $E_x^j \in \calE'_x$. First we consider the associated $D_x^j$. Notice how $A_x D_x^j = E^j_x$. 
    By the right local invertibility condition this means that $\restrB{x}{E_x^j}$ is well defined and equal to :
    $$\restrB{x}{E_x^j} \nameeq \restrA{x}{ D_x^j} = I_{D_x^j}.$$
    where the last equality is because $\{D_x^j\to E_x^j\}_{x\in\X,j\in J}$ characterizes $(\restriA{},A_{(-)})$.
    Moreover since both $A_{(-)}$ and $B_{(-)}$ are name-preserving by Prop.~\ref{prop:renaminginvarianceimpliesnp}, we have $V_{E_x^j} \nameeq V_{D_x^j}$. Since $B_{D_x^j} = B_{E_x^j}$, we even have $I_{E_x^j} \nameeq I_{D_x^j}$. This proves:
    $${\restriB{x}}({E_x^j}) \nameeq I_{E_x^j}$$
    Since by definition ${\restriB{x}}({E_x^j})\subseteq I_{E_x^j}$ this even enforces ${\restriB{x}}({E_x^j}) = I_{E_x^j}$, which finishes to prove $E_x^j\in \diskB{x}$.
    
    Now we prove \PA{$\calE'_x\subseteq \calE_x$}. We consider ${E_x^j}' \in \calE_x'$. 
    On the one hand we have:
    $${E_x^j}'=A_x B_x ({E_x^j}') =A_x({D_x^j}')$$
    On the other hand, since by right local invertibility we have ${D_x^j}'\in \graphvalidexplicit{\restriA{}}{x}$, 
    this graph can also be decomposed as $D_x^j\sqcup K$ with $I_K\cap I_{D_x^j}=\emptyset$ and it transforms as follow under the action of $A_{(-)}$:
    $$A_x ({D_x^j}') = A_x (D_x^j\sqcup K) = E_x^j\sqcup K$$
    Thus we have proven ${E_x^j}' = E_x^j\sqcup K$. Since \PA{$E_x^j \in \calE_x \subseteq \calE'_x$}, this proves $E_x^j = {E_x^j}'$ by unambiguity of the mutex set of cones $\calE':x\mapsto\{{E_x^j}'\}_{j\in J}$.

    Now we just have left to check that $ D_x^j = {D_x^j}'$ for all $i$. Consider $D_x^j\in \calD$. 
    One the one hand we have by left invertibility:
    $$ B_xA_x D_x^j = D_x^j$$
    On the other hand, $A_x  D_x^j = E_x^j = {E_x^j}'$, thus we have:
    $$B_xA_x D_x^j = B_x{E_x^j}' = {D_x^j}'$$
\end{proof}

\section{Commutative inverse}

\begin{lemma}[Pasts and their validity are preserved by $A_\omega$]\label{lemma : Non valid pasts are preserved}
    Let $A_{(-)}$ be a commutative and time-increasing local rule.
    Consider $G\in \sshift$ and $u\in \Past(G)$.
    \RX{}{Let $\tau_u=\{x\in\X \mid x\namesubseteq \sp{u} \text{ and }x\in\valid{G}{} \}.$\todo{To Pablo : Pour moi l'ensemble $\tau_u$ ne contient qu'un élément. Ce n'est pas le cas ?}}
    Consider $\omega\in \valid{G}{}$  such that for all $y\in\omega$, $y\neq \sp{u}$ \RX{$y\neq \sp{u}$}{$y\notin\tau_u$}.
    Consider $x=\sp{u}$ \RX{$x=\sp{u}$}{$x\namesubseteq u$}. 
    We have $u\in \Past(A_\omega G)$, and $x\in \valid{G}{} \Leftrightarrow x\in \valid{A_\omega G}{}$.
\end{lemma}
\begin{proof}
    By induction on $|\omega|$.
    This is obvious for $|\omega|=0$.
    Suppose this is true for $|\omega|=n$ and consider $\omega'=\omega y$.\\
    
    Since $u\in \Past(G)$, $y\neq \sp{u}$ \RX{$y\neq \sp{u}$}{$y\notin \tau_u$}, and due to the locality of $A_y$, we have we have $u\in \Past(A_y G)$. 
    Regarding the equivalence, 
    \begin{align*}
    x\notin \valid{G}{}
    &\Rightarrow A_x G = G \text{ by locality}\\
    &\Rightarrow A_yA_x G = A_y G\\
    &\Rightarrow A_y G = A_x A_y G \text{ by commutativity}\\
    &\Rightarrow u\in \Past(A_x A_y G)\\
    &\Rightarrow x\notin \valid{A_y G}{} \text{ by time-increasing.}
    \end{align*}
    \begin{align*}
    x\notin \valid{A_y G}{} &\Rightarrow A_x A_y G = A_y G \text{ by locality}\\
    &\Rightarrow A_y G= A_y A_x G \text{ by commutativity}\\
    &\Rightarrow u\in \Past(A_y A_x G)\\
    &\Rightarrow x\notin \valid{G}{}\text{ by time-increasing.}
    \end{align*}
    So, $x\in \valid{G}{}\Leftrightarrow x\in \valid{A_y G}{}$.\\
    \RX{}{It folllows that $\tau_{u\in A_y G}=\{x\in\X \mid x\namesubseteq \sp{u} \text{ and }x\in\valid{A_y G}{} \}.$ is the same as $\tau_{u \in G}$.}
    Applying the first part of the induction hypothesis, we have $u\in \Past(A_\omega A_y G)$ and so $u\in \Past(A_{\omega'})$ as required.\\
    Applying the second part of the induction, $A_y x\in \valid{A_y G}{}\Leftrightarrow x\in \valid{A_\omega A_y G}{}$. So, $x\in \valid{G}{} \Leftrightarrow x\in \valid{A_{\omega'} G}{}$ as required.
\end{proof}

\MC{
\begin{lemma}[Non valid pasts are preserved by $A_y$]\label{lemma : Non valid pasts are preserved 1 step}
    Let $A_{(-)}$ be a commutative and time-increasing local rule.
    Let $G\in \graphvalid{\restri{}}{y}$.
    Let $t.x$ \RX{$t.x$}{$t.z$} $\in \Past(G)$.
    \RX{}{For all $x\namesubseteq z$,} $G\notin \graphvalid{\restri{}}{x}$ implies $t.x$ \RX{$t.x$}{$t.z$} $\in A_y G\notin\graphvalid{\restri{}}{x}$.
\end{lemma}
\begin{proof}
    Since $t.x$ \RX{$t.x$}{$t.z$} $\in \Past(G)$ and $G\notin \graphvalid{\restri{}}{x}$, we have $A_yA_xG=A_yG$. By locality of $A_y$ we have $t.x$ \RX{$t.x$}{$t.z$} $\in V_{A_yA_xG}=V_{A_yG}$. Commutativity then gives us $t.x$ \RX{$t.x$}{$t.z$} $\in V_{A_xA_yG}=V_{A_yG}$. By the time-increasing condition this implies $A_y G\notin\graphvalid{\restri{}}{x}$.
\end{proof}

\begin{lemma}[Validity of commutativity]\label{lemma : 2valid}
    Let $A$ be a commutative and time-increasing local rule, and $G\in \graphvalid{\restri{}}{x}\cap \graphvalid{\restri{}}{y}$ be a graph.
    Then the sequence $yx$ is valid in $G$.
\end{lemma}
\begin{proof}
    We denote by $t.y$ \RX{$t.y$}{$t.z$} the vertex at position $y$ \RX{$y$}{$y\namesubseteq z$}.
    As there exists no path to $t.y$ \RX{$t.y$}{$t.z$} in $G$, the reachability condition of neighbourhood schemes implies that $t.y$ \RX{$t.y$}{$t.z$} $\notin \restr{x}{G}$. 
    By $\restri{}$-locality $A_x$ cannot modify $t.y$ \RX{$t.y$}{$t.z$} which implies $t.y$ \RX{$t.y$}{$t.z$} $\in V_{A_{x} G}$. On the other hand the time-increasing hypothesis tells us that if there exists $t'.y'$ \RX{$t'.y'$}{$t'.z'$} with $y=y'$ \RX{$y\namecap y'\neq\varnothing$}{$z\namecap z'\neq\varnothing$} $\in V_{A_xA_yG}$, then $t'>t$. Since we have $A_xA_y G=A_yA_x G$ this means that $A_y$ indeed modifies $t.y$ \RX{$t.y$}{$t.z$} in $A_x G$, thus we do have $A_x G\in \graphvalid{\restri{}}{y}$ which proves the validity of $yx$ in $G$.
\end{proof}}

\lemtimesymcommutation*\label{proof : time-symmetric commutation}

\begin{proof}
For a backward local rule the commutation equation must be true for any $H\in\ufut{u}\cap \ufut{v}$, \RX{$x=\sp{u}$}{$x\namesubseteq \sp{u}$} and \RX{$y=\sp{v}$}{$y\namesubseteq \sp{v}$}. Note how $\transpose{H} \in \upast{u}\cap \upast{v}$, we therefore have $A_{xy}\transpose{H}=A_{yx}\transpose{H}$. Thus, 
\begin{align*}
 B_{xy}H=(\transpose{~}\circ A_{x} \circ \transpose{~} \circ\transpose{~}\circ A_{y} \circ \transpose{~})H
 =\transpose{~}\circ A_{xy}\transpose{H}
 =\transpose{~}\circ A_{yx}\transpose{H}
 =B_{yx}H.
\end{align*}
\end{proof}

\proptwotwo*\label{proof : two-two}
\begin{proof}
$[$First part$]$ Consider $A_{(-)}$ commutative. 
Let us show that $B_{(-)}$ is two-two. Suppose $G\in \graphvalidexplicit{\restriA{}}{x}\cap\graphvalidexplicit{\restriA{}}{y}$.
We want $H\in\graphvalidexplicit{\restriB{}}{x}\cap\graphvalidexplicit{\restriB{}}{y}$ such that $G=B_{xy} H$. Let $H:=A_{xy} G$. \MC{Lem.\ref{lemma : Non valid pasts are preserved} tells us that $A_yG\in \graphvalid{\restri{}}{x}$. This implies, $H\in\graphvalidexplicit{\restriB{}}{x}$.} Using the commutativity of $A$, $H=A_{yx} G$ and so we also have that $H\in\graphvalidexplicit{\restriB{}}{y}$, and that $G=B_{xy} H$.\\
$[$Second part$]$ Consider $A_{(-)}$ a two-two, commutative, reversible local rule and let $B_{(-)}$ be its inverse. Take a graph $H\in \ufut{u}\cap \ufut{v}$,\RX{$x=\sp{u}$}{$x\namesubseteq \sp{u}$} and \RX{$y=\sp{v}$}{$y\namesubseteq \sp{v}$}. We will prove that $B_xB_y H=B_yB_x H$.\\
If $H\notin \graphvalidexplicit{\restriB{}}{x}\cup\graphvalidexplicit{\restriB{}}{y}$, we have $B_xB_y H=H=B_yB_x H$.\\
\MC{If $H\in \graphvalidexplicit{\restriB{}}{x}$ and $H\notin\graphvalidexplicit{\restriB{}}{y}$ the proof is more intricate. We first prove that $A_x H\in\graphvalidexplicit{\restriB{}}{y}\implies H\in\graphvalidexplicit{\restriB{}}{y}$. Suppose $H\in \graphvalidexplicit{\restriB{}}{y}$. We know by \bckwmaximality of $\sshift$, that there exists a graph $H'\in \maxi\!\sshift$ such that $H\sqsubseteq H' $. Since $A_{(-)}$ is two-two there exists $G\in\graphvalidexplicit{\restriA{}}{x}\cap \graphvalidexplicit{\restriA{}}{y}$ such that $A_{xy}G=H'$. Now we denote $Q=H'_{\overline{I_H}}$ and we derive using strong extensivity :
$$G=B_yB_x(H\sqcup Q) = B_y((B_xH)\sqcup Q)=(B_yB_xH)\sqcup Q$$
Since $B_yB_xH\sqsubseteq G$, we know that $x\in Past(B_yB_xH)$ and by reversibility we have $B_yB_xH \in \graphvalidexplicit{\restriA{}}{y}$. Since $A_{xy} (B_{yx} H) = G$, using Lem.\ref{lemma : Non valid pasts are preserved} we can even deduce $B_yB_xH \in \graphvalidexplicit{\restriA{}}{x}\cap \graphvalidexplicit{\restriA{}}{y}$. Then by commutativity of $A_{(-)}$ and Lem.\ref{lemma : Non valid pasts are preserved}, we deduce $H\in\graphvalidexplicit{\restriB{}}{y}$.\\
With this proof established, we know that in our case we necessarily have $H\notin\graphvalidexplicit{\restriB{}}{y}$, we can therefore deduce easily the commutation equation $B_xB_y H=B_xH=B_yB_x H$.}\\
If $H\notin \graphvalidexplicit{\restriB{}}{x}$ and $H\in\graphvalidexplicit{\restriB{}}{y}$, we apply the same reasoning.\\
Otherwise we have $H\in \graphvalidexplicit{\restriB{}}{x}\cap\graphvalidexplicit{\restriB{}}{y}$. Thus there exists $G\in\graphvalidexplicit{\restriA{}}{x}\cap\graphvalidexplicit{\restriA{}}{y}$ such that $H=A_{xy} G$.\\
We have, using the commutativity of $A_{(-)}$:
$$B_{yx}H=B_{yx}A_{xy}G=G=B_{xy}A_{yx}G=B_{xy}A_{xy}G=B_{xy}H.$$
Thus $B_{(-)}$ is commutative. By the first part it is also two-two.\\
The converse is true by symmetry.
\end{proof}

For the next proof we will need to consider some equivalence classes on graphs under renamings that leave some fixed positions $x$ and $y$ unchanged: $\renameclass{G}=\{RG\ |\ R(x)= x,\ R(y)=y\}$. Note how $A_{(-)}$ can be thought as acting directly on this equivalence class, because it is renaming invariant.
\lemcommutation*\label{proof : commutation}
\begin{proof}
    Consider the following sets of graphs : 
    $$\sshift_{\restriA{x}\cup\restriA{y}} = \{G_{\restriA{x}\cup \restriA{y}}\ |\ G\in \graphvalidexplicit{\restri{}}{x}\cap \graphvalidexplicit{\restri{}}{y}\}$$
    $$\sshift_{\restriB{y}\cup\restriB{x}} = \{H_{\restriB{y}\cup \restriB{x}}\ |\ H\in \graphvalidexplicit{\restriB{}}{x}\cap \graphvalidexplicit{\restriB{}}{y}\}$$
    Note that the corresponding set of equivalence classe $\renameclass{\sshift_{\restriA{x}\cup\restriA{y}}}$ is finite because since $\restriA{}$ is bounded by $k$ it contains only graphs with at most $2k$ internal vertices. Moreover $\renameclass{\sshift_{\restriB{y}\cup\restriB{x}}}$ has the same cardinality, because there exists a renaming invariant bijection from $\graphvalid{\restriA{}}{x}\cap\graphvalid{\restriA{}}{y}$ to $\graphvalid{\restriB{}}{x}\cap\graphvalid{\restriB{}}{y}$. 


    Since $A_{(-)}$ is reversible and commutative, and using Lem.\ref{lemma : Non valid pasts are preserved},  each graph $G\in \graphvalid{\restriA{}}{x}\cap\graphvalid{\restriA{}}{y}$ is mapped to a graph $H\in \graphvalid{\restriB{}}{x}\cap\graphvalid{\restriB{}}{y}$ by $A_{xy}$.  
    This means that the image of $A_{xy}$ over $\renameclass{\sshift_{\restriA{x}\cup\restriA{y}}}$ 
    is contained inside $\renameclass{\sshift_{\restriB{x}\cup\restriB{y}}}$. 
    Moreover by reversibility, we can define $B_{yx}$ on $A_{xy} \renameclass{\sshift_{\restriA{x}\cup\restriA{y}}}$ and it acts as an inverse, which proves that $A_{xy}$ is injective on ${\renameclass{\sshift_{\restriA{x}\cup\restriA{y}}}}$. Since $|{\renameclass{\sshift_{\restriA{x}\cup\restriA{y}}}}| = |\renameclass{\sshift_{\restriB{x}\cup\restriB{y}}}|$ this implies that this restricted version of $A_{xy}$ is also surjective.

    Consider $H\in \graphvalid{\restriB{}}{x}\cap \graphvalid{\restriB{}}{y}$, we can write this graph $H=H_{\restriB{x}\cup \restriB{y}}\sqcup H_{\overline{\restriB{x}\cup \restriB{y}}}$. Since $\renameclass{H_{\restriB{x}\cup \restriB{y}}}\in \renameclass{\sshift_{\restriB{x}\cup\restriB{y}}}$ there exists by surjectivity $\renameclass{G_{\restriA{x}\cup \restriA{y}}}$ such that $A_{xy}\renameclass{G_{\restriA{x}\cup \restriA{y}}} = \renameclass{H_{\restriB{x}\cup \restriB{y}}}$. Then the graph $G = G_{\restriA{x}\cup \restriA{y}}\sqcup H_{\overline{\restriB{x}\cup \restriB{y}}}$ is in $\graphvalid{\restriA{}}{x}\cap \graphvalid{\restriA{}}{y}$ and such that $A_{xy} G = H$. We have therefore proven that $A_{(-)}$ is two-two and can conclude the proof using Lem.~\ref{lemma : reducing hypothesis for commutation of B}.


\end{proof}

For the next proof we will need a notation for the rightmost sequence subtraction. Let $\omega\in \mathcal{X}^*$ and $\alpha \in \mathcal{X}^*$. We define recursively $(\omega \setminus \alpha)\in \mathcal{X}^*$ as :
\begin{equation*}
    \omega \setminus \alpha =
    \begin{cases}
        \omega & \text{if } |\alpha| = 0, \\
        \omega''\omega' & \text{if } |\alpha| = 1, \omega = \omega''\alpha\omega', \text{ and } \alpha \notin \omega'\\
        (\omega \setminus x) \setminus \alpha' & \text{if } \alpha = \alpha' x \text{ and } x \in \mathcal{X}.
    \end{cases}
\end{equation*}
For example if $\mathcal{X} = \{0,\dots ,9\}$, $\omega = 22159892$ and $\omega' = 28542$ we have 
$\omega\setminus \omega' = 2199$. It has been proven in Lem.$4$ of \cite{ArrighiDetRew}, that if both $\omega$ and $\omega'$ are valid in a graph $G$ sequences, then $\omega\setminus \omega'$ is a valid sequence in $A_{\omega'}G$.

We will also use the notation $\omega \cup \omega'$ to denote $(\omega \setminus \omega')\circ \omega'$. It has been proven in Cor.$1$ of \cite{ArrighiDetRew}, that even if by definition $\omega \cup \omega'$ is different from $\omega'\cup \omega$, on all graph $G$ in which both $\omega$ and $\omega'$ are valid, we have $A_{\omega \cup \omega'}G = A_{\omega' \cup \omega}G$ as long as $A_{(-)}$ is commutative.

\inversedet*\label{proof : the inverse dynamics has the same property}
\begin{proof}
    We fix :
    $$H = \B_{\omega_1 \cup \omega_2} G = \B_{\omega_2 \setminus \omega_1}\B_{\omega_1} G = B_{\omega_1 \setminus \omega_2}B_{\omega_2} G$$
    We take :
    $$\omega_1' =  (\omega_2 \setminus \omega_1)^{T}$$
    $$\omega_2' = (\omega_1 \setminus \omega_2)^{T}$$
    And we obtain :
    \begin{align*}
        A_{\omega_1'} H &= A_{(\omega_2 \setminus \omega_1)^{T}}  \B_{\omega_2 \setminus \omega_1}\B_{\omega_1} G\\
        &= \B_{\omega_1} G
    \end{align*}
    Symmetrically we also have $\B_{\omega_2} G = A_{\omega_2'} H$, which concludes the proof.
\end{proof}

\lemtwowaycommutation*\label{proof : two-way commutation}
\begin{proof}
    Let $G\in \graphvalid{\restriB{}}{x}^*$ such that $B_xG\in \graphvalid{\restriA{}}{y}^*$. Notice how by reversibility we have $B_xG \in \graphvalid{\restriA{}}{x}^*$. 
    This means that we can use commutativity of $A_{(-)}$ to get :
    \begin{align*}
        A_xA_yB_xG &= A_yA_xB_xG\\
      \Leftrightarrow  B_xA_xA_yB_xG &= B_xA_yA_xB_xG\\
      \Leftrightarrow  A_yB_xG &=B_xA_yG
    \end{align*}
    We get the second part of the the lemma by a symmetric reasoning, using commutativity of $B_{(-)}$.
\end{proof}

\proptwowayconsistency*\label{proof: two-way consistency}
\begin{proof}
    We start by proving by induction on the size of $\omega_1$ that :
    $$C_{\omega_1} G = A_{\omega_1^A} B_{\omega_1^B} G$$
    where $\omega_1^A$ and $\omega_1^B$ denote respectively a $Past$-valid and a $Fut$-valid sequence.
    Let us suppose it is true for sequences of lenght $n$. We take $x \omega$ a sequence of length $n+1$. Then we have 
    $$C_{x\omega} G = C_{x} A_{\omega^A} B_{\omega^B}G$$
    If $C_{\omega} G\in \graphvalid{\restriA{}}{x}^*$ this concludes the proof immediately. Otherwise $C_{\omega} G\in \graphvalid{\restriB{}}{x}^*$ and we can apply iteratively Lem.~\ref{lemma : two-way commutation} to obtain $C_{x\omega} G =  A_{\omega^A} B_x B_{\omega^B}G$.

    Now we can just apply Prop.\ref{prop : the inverse dynamics has the same property} to get two past-valid sequences $\omega_{1_B}',\omega_{2_B}'$ and a graph $G'$ such that :
    $$C_{\omega_1} G = A_{\omega_1^A} A_{\omega_{1_B}'} G'$$
    $$C_{\omega_2} G = A_{\omega_2^A} A_{\omega_{2_B}'} G'$$
\end{proof}

\section{Extensivity versus strong extensivity}


In the present paper we work with strongly extensive neighbourhoods, but in our previous paper \cite{ArrighiDetRew}, we work with extensive neighbourhoods. Let us recall the definitions.
\begin{definition}[(Strong) extensivity]
    A neighbourhood scheme $\restri{}$ is \emph{extensive} if and only if for all $G\in \graphvalid{\restri{}}{x}$, for all $H\in\sshift$, we have that  $G_\restri{x}\sqsubseteq H \sqsubseteq G$ implies $H_\restri{x} = G_\restri{x}$.
    It is {\em strongly extensive} if and only if for any $D\in \diskA{x}$ and $G\in \sshift$ such that $D\sqsubseteq G$, we have $G_{\restriA{x}} = D$.\\
\end{definition}
Notice that 
both notions of extensivity imply idempotency, i.e. the property that for all $G\in \graphvalid{\restri{x}}{}$, we have that $G_\restri{x}\in\graphvalid{\restriA{}}{x}$ and in fact $G_\restri{x}=(G_\restri{x})_\restri{x}$. Idempotency is in turn the key notion to work in the quantum setting of \cite{ArrighiQNT}.\\
In the present paper, we also work with the idea that neighbourhoods and thus local rules need be defined on \maxgraph graphs, i.e. so long as following direct edges does not hit the border. Reciprocally when we do hit a border, we do allow them to be undefined. This means that we do not think of borders as marking `definite walls'. Rather, we think of them marking `lack of information' beyond this point. I.e. our working subgraph is understood a being only a partial view of a larger graph. Knowledge of the larger graph would allow for the correct pursuit of the computation.\\
This interpretation leads us to formulate the following safety principle:
\begin{definition}[Safety principle] Consider $\restri{}$ a possibly non-strongly-extensive neighourhood scheme. It obeys the {\em safety principle} if and only if whenever $H\sqsubseteq G$, we have that 
$$G\notin\graphvalid{\restri{}}{x} \textrm{ or }G_\restri{x}\not\sqsubseteq H  \implies H\notin\graphvalid{\restri{}}{x}$$
i.e. if the larger graph is not enough for computing the neighbourhood, or if the neighbourhood it would compute is not yet present in the smaller graph, then the smaller graph is not enough for computing the neighbourhood.
\end{definition}
Let us show that extensivity and this safety principle, actually entail strong extensivity. 
\begin{proposition}[Extensivity and safety principle imply strong extensivity]
Consider $\restri{}$ an extensive neighourhood scheme, which obeys the safety principle and is such that $\diskA{}\subseteq\sshift$. This $\restri{}$ is strongly extensive.  
\end{proposition}
\begin{proof}
$[$Extensive implies idempodent$]$ Consider $G\in\sshift$, since $\sshift$ is closed under inclusion we have $G_\restri{x}\in\sshift$. Extensivity with $H=G_\restri{x}$ gives
$G_\restri{x}\sqsubseteq G_\restri{x} \sqsubseteq G$ implies $(G_\restri{x})_\restri{x} = G_\restri{x}$.
Since the LHS of the implication is trivially fulfilled, its RHS holds.\\
\noindent $[$Strong extensivity$]$
By contradiction suppose that we do not have strong extensivity. Then there exists $G,H\in\sshift$ such that 
$$H_\restri{x}\sqsubseteq G\textrm{ and }(G\notin \graphvalid{\restri{}}{x}\textrm{ or }G_\restri{x}\neq H_\restri{x}).$$
Let us show that $G_\restri{x}\sqsubseteq H_\restri{x}$ is impossible. Indeed, we would then have $G_\restri{x}\sqsubseteq H_\restri{x}\sqsubseteq G$, which by extensivity would give $G_\restri{x}=H_\restri{x}$, which we excluded.\\
So, $G_\restri{x}\notin \graphvalid{\restri{}}{x}$ or $G_\restri{x}\not\sqsubseteq H_\restri{x}$. We can apply the safety principle and see that $H_\restri{x}\notin \graphvalid{\restri{}}{x}$ which contradicts idempotency.
\end{proof}

\end{document}